\documentclass[cmp]{svjour}
\usepackage{graphics}
\begin{document}
\title{Ergodic properties of random billiards driven by thermostats.}
\author{Konstantin Khanin \and Tatiana Yarmola
\thanks{Partially supported by NSF MSPRF} \fnmsep \thanks{Department of Mathematics, University of Toronto, Bahen Centre,
40 St. George St., Room 6290, Toronto, Ontario, CANADA M5S 2E4. Tel: +1(416)978-3484, Fax:  +1(416)978-4107}%
}                     
\institute{University to Toronto \email{yarmola@math.toronto.edu}}
\date{Received: date / Accepted: date}
%
\communicated{name}
\maketitle
\begin{abstract}
We consider a class of mechanical particle systems interacting with thermostats. Particles move freely between collisions with disk-shaped thermostats arranged periodically on the torus. Upon collision, an energy exchange occurs, in which a particle exchanges its tangential component of the velocity for a randomly drawn one from the Gaussian distribution with the variance proportional to the temperature of the thermostat. In the case when all temperatures are equal one can write an explicit formula for the
stationary distribution. We consider the general case and show that there exists a unique
absolutely continuous stationary distribution. Moreover under rather mild conditions on the
initial distribution the corresponding Markov dynamics converges to the equilibrium with
exponential rate. One of the main technical difficulties is related to a possible overheating of
moving particle. However as we show in the paper non-compactness of the particle velocity
can be effectively controlled.
\end{abstract}
\section{Introduction}
\label{intro}
Rigorous derivations of macroscopic heat conduction laws from microscopic dynamics of mechanical models require good mixing properties and fast convergence of initial distributions to the invariant measure(s). For many such systems in non-equilibrium, e.g. when a system is coupled to two or more unequal heat reservoirs, pure existence of invariant measures is a nontrivial and open question due to non-compactness of the phase spaces. It is relatively easy to envision scenarios under which a particle freezes or heats up, which may push initial distributions towards zero or infinite energy levels and ultimately violate existence of the invariant measures. Tightness arguments are required in order to show that such scenarios occur with zero probability and invariant measures indeed exist. Even stronger controls are required to obtain mixing properties.

Mechanical particle systems coupled to heat reservoirs has seen a renewed interest and activity over the last decade. Various examples were introduced in \cite{Balint,Collet,Eckmann,Nonequilibrium,Klages_Nicolis_Rateitschak,Larralde_Leyvraz_MejiaMonasterio,Lin}. Rigorous results on ergodicity and absolute continuity of invariant measures assuming existence for some of these systems have been obtained \cite{Balint,Nonequilibrium,Eckmann,Yarmola}. These results, however, rely on sample paths with infinitesimally low probability of occurrence giving no control on times and rates and cannot be used to obtain mixing properties or even existence of the invariant measures. We present an example of a simple particle system interacting with thermostats for which we are able to control more regular sample paths.

The example is motivated as follows: Consider a system of $N$ non-interacting particles at various velocities bouncing elastically off the walls of a bounded domain. We assume for visualization purposes that $N$ is very large and the system is at temperature $T_0$ in the following sense: kinetic energies of the particles are distributed with the Gibbs distribution with parameter $\beta_0=\frac{1}{T_0}$, i.e. the probability that a given particle has kinetic energy $dE$ near $E$ is approximately $c e^{-\beta_0 E}dE$.

Let us introduce a thermostat into the system set at a different temperature $T_1 \ne T_0$ such that when a particle collides with the thermostat, an energy exchange occurs. That is, upon collision, the thermostat absorbs part of the particle's energy, which depends on the angle of the collision, and the particle acquires an energy $E$ from the thermostat drawn form Gibbs distribution with parameter $\beta_1$, where $\beta_1=\frac{1}{T_1}$. Over time, such a system is expected to settle at temperature $T_1$, i.e. the initial Gibbs distribution with parameter $\beta_0$ is expected to converge to the Gibbs distribution with parameter $\beta_1$. The questions of interest are whether the Gibbs distribution with parameter $\beta_1$ is indeed the unique invariant measure for the system to which all (or almost all) initial distributions converge, and if so, at which rate.

Now let us add another thermostat at yet a different temperature $T_2 \ne T_1$. In this case the system is \emph{not} in thermal equilibrium. Does an invariant measure exist for such a system? Is it unique and if so, do reasonable initial distributions converge to it and what are the rates of convergence? Similar questions may be asked in the presence of more than two thermostats at different temperatures.

In the absence of particle interactions the system with many particles is simply the product of one particle systems. The dynamics is described by a continuous-time Markov Process, which is deterministic apart from collisions with thermostats and upon a collision of a particle with a thermostat, a random perturbation occurs. This degeneracy of the Markov process allows to restrict the study to the discrete time dynamics on the collision manifold.

For the resulting discrete-time Markov chain we show that, under certain geometric assumptions, there exists an invariant measure and it is unique (ergodic), absolutely continuous with respect to Lebesgue measure, and mixing with exponential rates. We also conclude that reasonable initial distributions converge to the invariant measure exponentially fast with control on the rates. It follows that for the original Markov process there exists an invariant measure and this measure is unique (ergodic) and absolutely continuous. Mixing and convergence of initial distributions to the invariant measure for the Markov process do not follow directly since under some scenarios particles may move extremely slow. Though ergodicity guarantees that they will eventually speed up, bounds on the times must be obtained in order to show mixing. We leave the investigation of mixing properties of the Markov process for future work.

The proof of mixing properties in the discrete case uses general state Markov chain machinery, in particular, Harris' Ergodic Theorem \cite{Hairer,Meyn}. The theorem requires two things: to produce a non-negative function $V$ on the phase space which, on average, decreases geometrically under the push forwards of the dynamics, and, given such a $V$, to show minorization or Doeblin's condition on certain level set of $V$. The first condition guarantees that the dynamics enters the 'center' of the phase space, a certain level set of $V$, with good control on the rates; and, once at the 'center', coupling is guaranteed by the minorization condition. Those two conditions imply existence and uniqueness of the invariant measure with exponential mixing rates and exponential convergence of reasonable initial distributions to that invariant measure.

We believe that most of our results can be generalized to the case of thermostats of a general smooth convex shapes rather than disks. It is also interesting and tempting to extend the methods of this paper to the $3$-dimensional case. We are planning to address these and other open questions in the near future.

We describe precise settings in section \ref{sect:settings}, then state the results and outline the proofs in section \ref{sect:results}. In section \ref{sect:potential} we describe the potential $V$, and in section \ref{sect:minorization} we establish the minorization condition.

\section{Settings} \label{sect:settings}

\begin{figure}
  \centering
  \resizebox{0.75\textwidth}{!}{
  \includegraphics{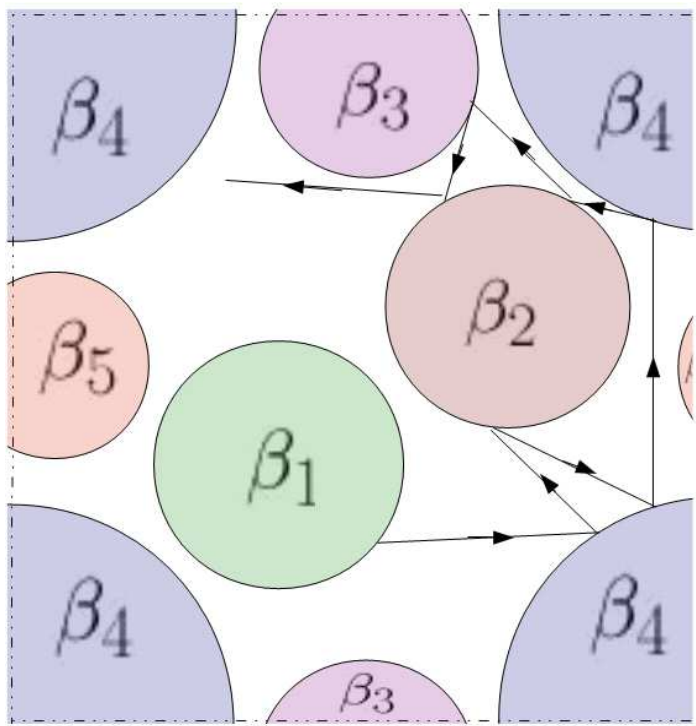}}
  \caption{Geometric Configuration}\label{fig: configuration}
\end{figure}

\begin{figure}
  \centering
  \resizebox{0.75\textwidth}{!}{
  \includegraphics{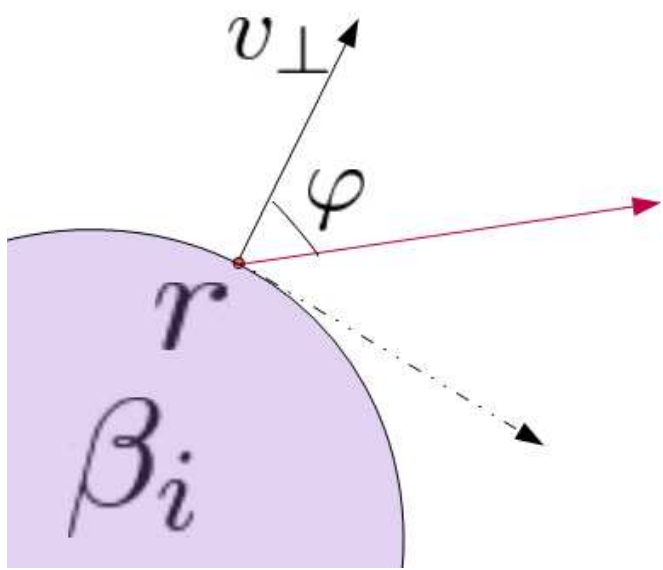}}
  \caption{Discrete Dynamics Coordinates}\label{fig:discrete dynamics coordinates}
\end{figure}

Let
$\Gamma = T^2 \setminus \cup_{i=1}^p D_i$ be a bounded horizon billiard table with finitely many non-intersecting circular obstacles $D_1, \cdots, D_p$ having radii $R_1, \cdots, R_p$. An example of such configuration is shown in Fig. \ref{fig: configuration}. Particles move around in $\Gamma$. Apart from collisions with the boundary $\partial \Gamma$, the particles move freely with constant velocities and do not interact with each other. Each obstacle $D_i$, $1 \leq i \leq p$, is assigned certain parameter $\beta_i$ which represents the inverse temperature and plays a role of a thermostat.

Upon a collision of a particle with a thermostat $D_i$, let $v^-=(v_\perp^-, v_t^-)$ be the decomposition of the particle's velocity into the normal and the tangential components with respect to the boundary of the thermostat $\partial D_i$. After the collision, the normal component of the velocity changes sign, $v_\perp^+=-v_\perp^-$, while the tangential component is absorbed by the thermostat and a new tangential component $v_t^+$ is drawn from the distribution $\sqrt{\frac{\beta_i}{\pi}} e^{-\beta_i v_t^2} dv_t^+$. The outgoing velocity has a decomposition $v^+=(v_\perp^+,v_t^+)$. This type of energy exchange was used in \cite{Lin}.

Since particles do not interact with each other, we can reduce our discussion to studying the system of only one particle. The Phase Space of such a systems is

$$\tilde{\Omega}=\{(x,v):x \in \Gamma, v \in R^2\}/\sim$$
where $\sim$ is an identification of points on the collision manifold that correspond to incidence and reflection with the choice of the old, incoming, velocities $v^-=(v_\perp^-,v_t^-)$.

Then the dynamics is described by a Markov Process $\tilde{\Phi}_\tau$ on $\tilde{\Omega}$, which is deterministic between collisions with thermostats with random kicks at the moments of collisions with $\partial \Gamma$.

As usual for such systems, it is convenient to pass to a discrete dynamics on the boundary $\partial \Gamma$ of $\Gamma$. Note that at the moment of collision we chose to keep old, incoming velocities. Given such a choice, to describe the discrete dynamics it is enough to keep track of the sequence of collision points $r_n \in \partial \Gamma$ (parameterized by arc-length) and the absolute value of the normal velocity $v_\perp(n)$ at the moment of collision. Using randomly generated tangential velocity $v_t(n)$ drawn from the probability distribution $\sqrt{\frac{\beta}{\pi}}e^{-\beta v_t^2}d v_t$, one determines the outward velocity and, hence, the next point of collision $r_{n+1}$ and the next normal velocity $v_\perp(n+1)$. This procedure defines a Markov chain $\Phi$ on $\Omega=\partial \Gamma \times [0,\infty)$.

For our future analysis it would be more convenient to use a geometric description similar to one used in billiards. Namely, instead of the random variable $v_t(n)$, we shall consider the random variable $\varphi_n$, an angle relative to the normal vector at the collision point at which a particle leaves after collision $n$. It is an easy calculation to show that $\varphi_n$ is drawn from the distribution $\rho_{v_\perp(n)}(\varphi_n) d \varphi_n=\sqrt{\frac{\beta_i}{\pi}} \frac{v_\perp(n)}{\cos^2(\varphi_n)} e^{-\beta_i v^2_\perp(n) \tan^2(\varphi_n)} d \varphi_n$, which depends on $\beta_i$ associated with $\partial D_i \ni r_n$ and $v_\perp(n)$. We shall also use notation $\varphi_{n}'$ for the incoming angle at the point of collision $r_{n+1}$. Obviously, $r_{n+1}$, $v_\perp(n+1)$, and $\varphi_{n}'$ are completely determined by $r_n$, $v_\perp(n)$, and $\varphi_n$. Note that $v_\perp(n+1)=\frac{\cos(\varphi_n')}{\cos(\varphi_n)}v_\perp(n)$.

Denote the transition probability kernel of $\Phi$ by $\mathcal{P}((r,v_\perp),\cdot)$, i.e. $\mathcal{P}((r,v_\perp), A)=P(\Phi_n \in A| \Phi_{n-1}=(r, v_\perp))$. We will also use the notation $\mathcal{P}^*$ and $\mathcal{P}_*$ for the operators defined on both the set of bounded measurable function and the set of measures of finite mass by
$$(\mathcal{P}^*f)(x,v_\perp)=\int_\Omega f(r',v_\perp')\mathcal{P}((r,v_\perp),d(r',v_\perp'))$$
$$(\mathcal{P}_*\mu)(A)=\int_\Omega \mathcal{P}((r,v_\perp),A)\mu(d(r,v_\perp))$$
Note that the transition probabilities are degenerate: $\mathcal{P}((r,v_\perp),\cdot)=\mathcal{P}_* \delta_{(r,v_\perp)}$ is supported on a family of one dimensional curves in the two dimensional phase space $\Omega$.

Our interest lies in investigating the questions of existence, uniqueness, absolute continuity w.r.t. Lebesgue measure $m$, ergodicity, and mixing properties of the stationary (invariant) measures for the Markov Chain $\Phi$ and the associated Markov process $\tilde{\Phi}_\tau$.

\section{Results} \label{sect:results}
We show existence of the invariant measures in the non-equilibrium situation simultaneously with the geometric ergodicity for the Markov chain $\Phi$.

The main result is given by the following Theorem:

\begin{theorem} \label{thm}
The Markov Chain $\Phi$ admits a unique absolutely continuous invariant probability measure $\mu$. Furthermore, there exist a non-negative function $V$ on $\Omega$, as well as constants $C>0$ and $\tilde{\gamma} \in (0,1)$ such that
$$\sup\limits_{A \subset \Omega}|\mathcal{P}^n((r,v_\perp),A)-\mu(A)| \leq C \tilde{\gamma}^n (1+V(r,v_\perp))$$
\end{theorem}

\begin{corollary} \label{coro: mixing}
The invariant measure of the Markov chain $\Phi$ is exponentially mixing, i.e. for any Borel $A,B \in \Omega$,
$$\sup_{A \subset \Omega} |\int\limits_B \mathcal{P}^n((r,v_\perp),A) d\mu -\mu(B)\mu(A)| \leq C \tilde{\gamma}^n \int\limits_B(1+V(r,v_\perp))d \mu \leq D e^{-\alpha n},$$
for $\alpha=-\ln(\tilde{\gamma})$ and $D =C \int\limits_\Omega(1+V(r,v_\perp))d \mu < \infty$.
\end{corollary}

Corollary \ref{coro: mixing} follows from Theorem \ref{thm} once $\int\limits_\Omega(1+V(r,v_\perp))d \mu < \infty$ is established. This is one of the results of the Theorem 3.2 in \cite{Hairer}.

\begin{corollary} \label{coro: convergence}
If the probability measure $\nu$ on $\Omega$ satisfies $\int_\Omega (1+V(r,v_\perp))d\nu < \infty$, then
$$\|\mathcal{P}^n_* \nu - \mu \| \leq \tilde{C} \tilde{\gamma}^n \int_\Omega (1+V(r,v_\perp))d\nu,$$
where $\|\cdot\|$ is a bounded variation norm.
\end{corollary}

By constructing a suspension flow over $\Phi$ we obtain existence, uniqueness and absolute continuity of the invariant measure for the Markov process $\tilde{\Phi}_\tau$ both for one and for many particle systems.

\begin{corollary} \label{coro}
There exists a unique (ergodic) absolutely continuous invariant probability measure for the Markov process $\tilde{\Phi}_\tau$
\end{corollary}

\begin{remark}
In the equilibrium case, i.e. when $\beta_1=\beta_2=\cdots=\beta_p=\beta$ the invariant measures can be written down explicitly.
\end{remark}

\begin{lemma} \label{lemma: equilibrium}
The measure $\mu$ with density $$d \mu = \frac{2 \beta}{|\partial \Gamma|} v_\perp e^{- \beta v^2_\perp} d v_\perp d r$$ is invariant for the Markov chain $\Phi$.
\end{lemma}

We will prove Lemma \ref{lemma: equilibrium} at the end of section  \ref{sect:minorization}.

\subsection{Idea of Proof of Theorem \ref{thm}}
To prove Theorem \ref{thm} we will use the following formulation of the Harris ergodic theorem for Markov chains also referred to as Geometric Ergodicity Theorem.

\begin{theorem}\cite{Hairer}
Assume
\begin{description}
\item[Potential Condition].\\ There exists a function $V: \Omega \to [0,\infty)$, $K>0$ and $\gamma \in (0,1)$ such that
$$\mathcal{P}^*V(r,v_\perp) \leq \gamma V(r,v_\perp)+K$$
for all $(r,v_\perp) \in \Omega$
and
\item[Minorization Condition].\\
There exists a probability measure $\nu$ supported on $\mathcal{C}$, $N$ and $\eta_N \in (0,1)$ such that
$$\inf\limits_{(r,v_\perp) \in \mathcal{C}}\mathcal{P}^{N}((r,v_\perp),\cdot) \geq \eta_N\nu(\cdot),$$
where $\mathcal{C}=\{(r,v_\perp) \in \Omega: V(r,v_\perp)\leq S\}$ for some $S>2K/(1-\gamma)$ where $K$ and $\gamma$ are the constants from the Potential Condition. In addition, to ensure aperiodicity, we require that the same holds for $N+1$ and some $\eta_{N+1} \in (0,1)$.
\end{description}
Then $\Phi$ admits a unique invariant measure $\mu$. Furthermore, there exist $C>0$ and $\gamma \in (0,1)$ such that for all $(r,v_\perp) \in \Omega$
$$\|(\mathcal{P}^{*n} f)(r,v_\perp) -\mu(f)(r,v_\perp)\|\leq C \tilde{\gamma}^n  \|f-\mu(f) \|,$$
for all $f$ such that $\|f\|<\infty$, where $\|f\|=\sup\limits_{(r,v_\perp)} \frac{|f(r,v_\perp)|}{1+V(r,v_\perp)}$
\end{theorem}

In particular, for any $A \subset \Omega$, let $f=\left\{
                                                  \begin{array}{ll}
                                                    1, & (r,v_\perp) \in A \\
                                                    0, & (r,v_\perp) \not \in A
                                                  \end{array}
                                                \right.
$, then
$$\frac{|\mathcal{P}^n((r,v_\perp),A)-\mu(A)|}{V(r,v_\perp)+1} \leq \sup\limits_{(r,v_\perp)}\frac{|\mathcal{P}^{*n}(f)-\mu(f)|}{V(r,v_\perp)+1} \leq C \tilde{\gamma}^n$$
so that
$$\sup\limits_{A \subset \Omega}|\mathcal{P}^n((r,v_\perp),A)-\mu(A)| \leq C \tilde{\gamma}^n (1+V(r,v_\perp)),$$
which is exactly a conclusion of Theorem \ref{thm}. We chose a weaker formulation of the theorem in order to make the presentation more intuitive.

To prove Theorem \ref{thm} we need to show that the Markov Chain $\Phi$ satisfies the conditions of the Harris Ergodic Theorem and that the invariant measure we obtain is absolutely continuous with respect to the Lebesgue measure on $\Omega$. We construct the potential $V$ in section \ref{sect:potential} and prove the minorization condition on a level set of $V$ and absolute continuity in section \ref{sect:minorization}.

\section{Potential V} \label{sect:potential}

\begin{proposition} \label{prop:V}
There exists a function $V:\Omega \to [0,\infty)$ and constants $K>0$ and $\gamma \in (0,1)$ such that
\begin{equation} \label{eqn:potential}
\mathcal{P}^*V(r,v_\perp)= \int\limits^{\pi/2}_{-\pi/2} V(r',v'_\perp) \rho_{v_\perp}(\varphi; \beta_i) d \varphi \leq \gamma V(r,v_\perp) + K, \; \; \; \forall (r,v_\perp) \in \Omega,
\end{equation}
where
$$\rho_{v_\perp}(\varphi; \beta_i)=\sqrt{\frac{\beta_i}{\pi}} \frac{v_\perp}{\cos^2(\varphi)} e^{- \beta_i v^2_\perp \tan^2(\varphi)}$$
and $r'$ and $v'$ are the position and the normal velocity at the next collision given $r$, $v_\perp$, and $\varphi$.
\end{proposition}

The density $\rho_{v_\perp}(\varphi; \beta_i)$ implicitly depends on $r$ through the inverse temperature of the corresponding thermostat. Since we have only a finite number of thermostats, the $\beta_i$ are bounded above and below and do not play an important role in the asymptotic analysis. To simplify the notation we will write $\rho_{v_\perp}(\varphi)$ instead of $\rho_{v_\perp}(\varphi; \beta_i)$.

\subsection{Heuristics}
The inequality (\ref{eqn:potential}) ensures that the values of $V$ at the random images of $(r,v_\perp)$ are, on average, smaller than $V(r,v_\perp)$ when $V(r,v_\perp)$ is large, and, on average, smaller than a constant when $V(r,v_\perp)$ is small. This implies the dynamics enters the \textquoteleft center' of the phase space, represented by some level set of $V$, regularly with tight control on the length of excursions from the \textquoteleft center' \cite{Hairer,Meyn}. Since we expect that particles reach very low or very high velocities with very low probabilities and would like the system to satisfy the minorization condition on the center, a natural physically meaningful candidate for the center would be a set of states with moderate $v_\perp$ velocities, e.g. $\mathcal{C}=\{(r,v_\perp): v^{\min}_\perp < v_\perp < v^{\max}_\perp\}$ for some $0<v^{\min}_\perp<v^{\max}_\perp<\infty$. Thus we are looking for a potential $V$ that takes $O(1)$ values in $\mathcal{C}$ and blows up as $v_\perp \to 0$ and as $v_\perp \to \infty$. This is essential in order to satisfy the minorization condition.

First let us investigate heuristically the mechanism that ensures that a potential of this kind should, in principle, satisfy the inequality in Prop. \ref{prop:V}. Assume for simplicity that $V$ depends on $v_\perp$ only and $\mathcal{C}$ is a level set of $V$, i.e. $\mathcal{C}=\{(r,v_\perp): v^{\min}_\perp < v_\perp < v^{\max}_\perp\}=\{(r,v_\perp): V(v_\perp)<A\}$ for some $A$. We are interested in the change of potential we can expect if we start with low, moderate and high $v_\perp$ and iterate one step forward along some likely random trajectory. The velocity of the particle between collisions with thermostats is $v=v_\perp/\cos(\varphi)$, where $\varphi$ is the angle drawn from $\rho_{v_\perp}(\varphi)d\varphi$, and the normal velocity upon collision is $v'_\perp = v \cos(\varphi')$, where $\varphi'$ is the angle of incidence of the particle that originated at $(r,\varphi)$.

\begin{itemize}
\item If $v_\perp$ is very small, with large probability $v_t$ is the main contribution to the overall velocity $v$ and $v'_\perp \approx \cos(\varphi') v_t$. Again, with large probability, $\varphi'$ is bounded away from $\pm \frac{\pi}{2}$ so that $v'_\perp$ is of moderate range. Thus with large probability the value of $V$ drops, i.e. $V(v'_\perp)<V(v_\perp)$.
\item If $v_\perp$ is $O(1)$, with large probability the $v_t$ contribution is \textquoteleft comparable' to $v_\perp$, so  $v$ is of the same order of magnitude as $v_\perp$. Again, with large probability $\varphi$ and $\varphi'$ are bounded away from $\pm \frac{\pi}{2}$, and thus $v'_\perp$ is of the same order of magnitude as $v_\perp$ and $V(v'_\perp)=O(1)$.
\item If $v_\perp$ is very large: with large probability $v_t$ is negligible and $\varphi \approx 0$. Thus  $v'_\perp \approx \cos(\varphi')v_\perp < v_\perp$ provided $\cos(\varphi')$ is bounded away from $0$ i.e. with large probability $V(v'_\perp) < V(v_\perp)$.
\end{itemize}

We conclude that, in principle, a potential $V$ taking $O(1)$ values in $C$ and tending to infinity as $v_\perp \to 0$ and as $v_\perp \to \infty$ should satisfy the inequality in Prop. \ref{prop:V}. The escapes to infinity should be fast enough to ensure large enough drops of averaged value, though not too fast to ensure integrability of the left hand side. There are also several potentially dangerous geometric locations: when $v_\perp$ is very small near locations $r \in \partial \Gamma$ such that $\varphi=-\varphi' = \pm \frac{\pi}{2}$ and when $v_\perp$ is very large near $r \in \partial \Gamma$ such that $\varphi=\varphi'=0$. The heuristic large probability arguments fail for those. In the next subsection we will present a potential $V$ that depends on $v_\perp$ and show that it satisfies the inequality of Proposition \label{prop:V} for all small, moderate $v_\perp$, and large $v_\perp$, even in the aforementioned dangerous geometric locations.

\subsection{Construction of the potential $V$}

A reasonable guess for $V(v_\perp)$ is
\begin{equation} \label{eqn:potential v_perp}
V(v_\perp) =\left\{
               \begin{array}{ll}
                 e^{\epsilon v_\perp^2}, & v_\perp>v^{\max}_\perp; \\
                 v^{-a}_\perp, & v_\perp<v^{\min}_\perp;\\
                 A, & v^{\min}_\perp \leq v_\perp \leq v^{\max}_\perp.
               \end{array}
             \right.
\end{equation}
where $\epsilon < \beta_{\min} = \min\{\beta_1, \cdots, \beta_p\}$ and $0<a<2$. Let us fix $a=-1$. We will fix $\epsilon$ later. We will also assume that $e^{\epsilon(v_\perp^{\max})^2}=(v_\perp^{\min})^{-1}=A$ so that $V$ is continuous.

We are interested in establishing that for some $\gamma>1$ and $K<\infty$

\begin{equation} \label{eqn: 3 part split}
\int\limits_{-\frac{\pi}{2}}^{\frac{\pi}{2}} V(v_\perp') \rho_{v_\perp}(\varphi) d \varphi
\end{equation}
$$\leq
\int\limits_{-\frac{\pi}{2}}^{\frac{\pi}{2}} [\frac{\cos(\varphi')v_\perp}{\cos(\varphi)}]^{-1} \rho_{v_\perp}(\varphi) d \varphi + \int\limits_{-\frac{\pi}{2}}^{\frac{\pi}{2}} \exp(\epsilon[\frac{\cos(\varphi')v_\perp}{\cos(\varphi)}]^{2}) \rho_{v_\perp}(\varphi) d \varphi
+ A\int\limits_{-\frac{\pi}{2}}^{\frac{\pi}{2}} \rho_{v_\perp}(\varphi) d \varphi $$
$$\leq \gamma V(v_\perp)+K$$

\subsubsection{Controlling small $v_\perp$}
In this subsection we are going to check the inequality for a potential (\ref{eqn:potential v_perp}) given $(r, v_\perp)$, where $v_\perp < v^{\min}_\perp$ for some $v_\perp^{\min}$ small enough. When the random map $\Phi$ is applied, we first draw $\varphi$ from the distribution with density $\rho_{v_\perp}(\varphi)$ and then jump to $(r',\varphi', v_\perp')$, where $(r',\varphi')$ is the image of $(r,\varphi)$ under the billiard map and $v_\perp'=\frac{\cos(\varphi')}{\cos(\varphi)}v_\perp$. As $\varphi$ sweeps from $-\frac{\pi}{2}$ to $\frac{\pi}{2}$, the billiard map image $(r',\varphi')$ sweeps through several thermostats (we only count those with positive measure of $(r',\varphi')$-images). These thermostats can be divided into two groups: \textquoteleft outer' thermostats, roughly speaking those that intersect two rays originating at $(r,\pm\frac{\pi}{2})$ and \textquoteleft inner' thermostats, all the rest. Of course, some thermostats in this description may repeat: nevertheless we may classify the thermostat differently depending at which angle we are looking at it.

Figure \ref{fig:cosfrac} gives an idea of how $\frac{\cos(\varphi')}{\cos(\varphi)}$ behaves as $\varphi$ sweeps from $-\frac{\pi}{2}$ to $\frac{\pi}{2}$ in a typical case. Each curve represents a thermostat. The figure does not specify which thermostat blocks another, so there is an overlay of the values and we will be integrating over all the overlays since we do not make any assumptions on geometry; of course, only one branch works for each value of $\varphi$ in the real setting.

Fix $0<\gamma_s<1$. We would like to split the first integral in (\ref{eqn: 3 part split}) into two parts depending whether $\frac{\cos(\varphi')}{\cos(\varphi)}<\frac{1}{\gamma_s}$ or $\frac{\cos(\varphi')}{\cos(\varphi)} \geq \frac{1}{\gamma_s}$.

$$\int\limits^{\pi/2}_{-\pi/2} [\frac{\cos(\varphi') v_\perp}{\cos(\varphi)}]^{-1} \sqrt{\frac{\beta_i}{\pi}} \frac{v_\perp}{\cos^2(\varphi)} e^{- \beta_i v^2_\perp \tan^2(\varphi)} d \varphi$$

$$= v_\perp^{-1} \int\limits_{\frac{\cos(\varphi')}{\cos(\varphi)} \geq \frac{1}{\gamma_s}} [\frac{\cos(\varphi')}{\cos(\varphi)}]^{-1}  \sqrt{\frac{\beta_i}{\pi}} \frac{v_\perp}{\cos^2(\varphi)} e^{- \beta_i v^2_\perp \tan^2(\varphi)} d \varphi$$

$$+ \int\limits_{\frac{\cos(\varphi')}{\cos(\varphi)}<\frac{1}{\gamma_s}} [\frac{\cos(\varphi')}{\cos(\varphi)}]^{-1} \sqrt{\frac{\beta_i}{\pi}} \frac{1}{\cos^2(\varphi)} e^{- \beta_i v^2_\perp \tan^2(\varphi)} d \varphi$$

\begin{equation} \label{eqn:gamma split}
\leq \gamma_s v_\perp^{-1}+ \sqrt{\frac{\beta_{\max}}{\pi}} \int\limits_{\frac{\cos(\varphi')}{\cos(\varphi)}<\frac{1}{\gamma_s}} [\frac{\cos^2(\varphi')}{\cos^2(\varphi)}]^{-\frac{1}{2}} \frac{d \varphi}{\cos^2(\varphi)},
\end{equation}

where $\beta_{\max}=\max\{\beta_1, \cdots, \beta_p\}$.

To estimate the second integral in (\ref{eqn:gamma split}), we need to express $\cos(\varphi')$ in terms of $\varphi$ for each thermostat in range. Given a thermostat $D_i$ let $\varphi_0$ be such that corresponding $\varphi_0'=-\frac{\pi}{2}$ if we assume that all other thermostats are invisible, i.e. the particle can pass through them. See Fig. \ref{fig: varphi prime}.
For some of the \textquoteleft outer' thermostats such $\varphi_0$ will not exist: in this case we have to choose $\varphi_0$ with corresponding $\varphi_0'=\frac{\pi}{2}$ and perform a mirror-image parametrization of $\cos(\varphi')$. The computation is essentially the same and thus omitted.

\begin{figure}
  \centering
  \resizebox{0.75\textwidth}{!}{
  \includegraphics{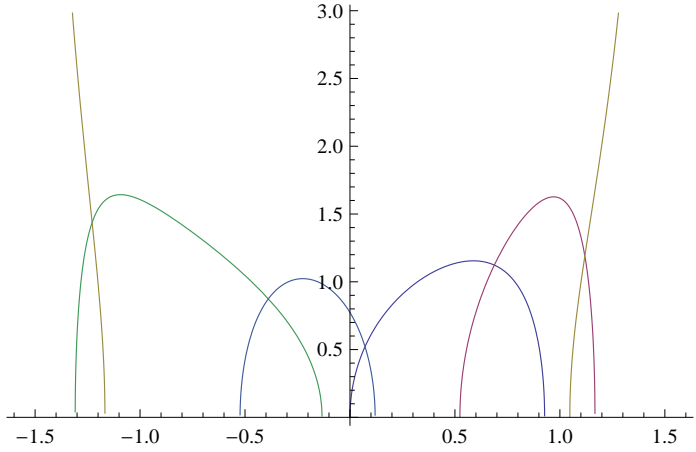}}
  \caption{$\frac{\cos(\varphi')}{\cos(\varphi)}$ for all visible scatterers}\label{fig:cosfrac}
\end{figure}

\begin{figure}
  \centering
  \resizebox{0.75\textwidth}{!}{
  \includegraphics{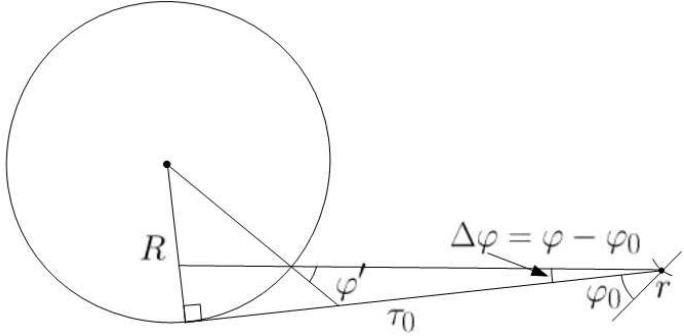}}
  \caption{Parametrization of $\varphi'$}\label{fig: varphi prime}
\end{figure}

By sine theorem (see Fig \ref{fig: varphi prime})

$$\frac{R}{\cos(\varphi-\varphi_0)}=\frac{R-\tau_0 \tan(\varphi - \varphi_0)}{\sin(-\varphi')}$$

Thus

$$\sin(-\varphi')=\cos(\varphi - \varphi_0)-\frac{\tau_0}{R}\sin(\varphi-\varphi_0)=\cos(\varphi)\cos(\varphi_0)+\sin(\varphi)\sin(\varphi_0)$$
$$-\frac{\tau_0}{R}\sin(\varphi)\cos(\varphi_0)+\frac{\tau_0}{R}\cos(\varphi)\sin(\varphi_0)=\frac{1+w w_0 - \frac{\tau_0}{R}w + \frac{\tau_0}{R}w_0}{\sqrt{w^2+1} \sqrt{w^2_0 + 1}}
$$
here we changed variables $w=\tan(\varphi)$ with $w_0=\tan(\varphi_0)$.
$$\frac{\cos^2(\varphi')}{\cos^2(\varphi)}=\frac{(w^2+1)(w^2_0+1)-(1+w w_0 - \frac{\tau_0}{R}w + \frac{\tau_0}{R}w_0)(1+w w_0 - \frac{\tau_0}{R}w + \frac{\tau_0}{R}w_0)}{w^2_0+1}=$$
$$=\frac{1+2\frac{\tau_0}{R}w_0-\frac{\tau_0^2}{R^2} }{w^2_0+1}w^2-2\frac{w_0-\frac{\tau_0}{R}-\frac{\tau_0^2}{R^2}w_0+\frac{\tau_0}{R}w^2_0}{w^2_0+1}w+\frac{w^2_0-2\frac{\tau_0}{R}w_0-\frac{\tau_0^2}{R^2}w^2_0}{w^2_0+1}$$
$$=:a w^2 - 2b w + c$$

For \textquoteleft inner' thermostats, which we define as thermostats with corresponding $a<0$, the above expression is a concave down parabola. We are interested in the values between the roots, where the smaller of the roots is $w_0$. The maximum value at each parabola is at least $1$ since for some $\varphi$, the particle hits the thermostat radially, i.e. with $\cos(\varphi')=1$, so that $\frac{\cos^2(\varphi')}{\cos^2(\varphi)} = \frac{1}{\cos^2(\varphi)} \geq 1$. Let $l$ be the maximal value of the parabola if  $aw^2-2bw+c=\frac{1}{\gamma_s^2}$ has no real solutions and $l=\frac{1}{\gamma_s^2}$ otherwise.
Denote the smaller or the only root of $aw^2-2bw+c=l$ by $w_1$. Then

$$\frac{l(w-w_0)}{(w_1-w_0)} \leq aw^2-2bw+c \leq \frac{1}{\gamma_s^2} \;\;\; on  \;\;\; [w_0,w_1] \;\;\; and$$

$$\int\limits_{w_0}^{w_1} (aw^2-2bw+c)^{-1/2}dw \leq \int\limits_{w_0}^{w_1} (\frac{l(w-w_0)}{(w_1-w_0)})^{-1/2}dw=2(w_1-w_0) l^{-\frac{1}{2}}$$

It is left to estimate $(w_1-w_0)$. Simple computation using the values of coefficients $a$, $b$, and $c$ ensures that $\sqrt{b^2-ac}=\frac{\tau_0}{R}$ and that $$(w_1-w_0)=\frac{\sqrt{\frac{\tau_0^2}{R^2}+a l}-\frac{\tau_0}{R}}{a}$$
has no extremum points for $a \in [-\frac{\tau_0^2}{l R^2}, 0]$. Thus

$$(w_1-w_0) \leq \max \{\frac{R l}{\tau_0}, \lim\limits_{a \to 0} (w_1-w_0)\} = \frac{R l}{\tau_0}$$

We conclude that

$$\int\limits_{w_0}^{w_1} (aw^2-2bw+c)^{-1/2}dw \leq 2\frac{R l}{\tau_0} l^{-\frac{1}{2}} \leq \frac{2R_{\max}}{\tau_{\min} \gamma_s},$$

where $R_{\max}=\max\{R_1, \cdots, R_p\}$ and $\tau_{\min}$ is the minimum distance of flight.

The situation is clearly symmetric for another set of roots of $aw^2-2bw+c=0$ and $aw^2-2bw+c=l$.

For the \textquoteleft outer' thermostats, when $a \geq 0$, $aw^2-2bw+c$ is a concave up parabola or a line with slope $2\frac{\tau_0}{R}$ (case $a=0$) and we are interested in integrating values with $w>w_0$. The derivative of $aw^2-2bw+c$ at $w_0$ is $2\frac{\tau_0}{R}$ and $2\frac{\tau_0}{R}(w-w_0) \leq aw^2-2bw+c $ on $[w_0,\infty]$. Let $w_1=w_0+\frac{R}{2 \tau_0 \gamma_s^2}$. Then $aw^2-2bw+c > \frac{1}{\gamma_s^2}$ on $(w_1,\infty]$ and

$$\int\limits_{w_0}^{w_1} (aw^2-2bw+c)^{-1/2}dw \leq \int\limits_{w_0}^{w_1} (2\frac{\tau_0}{R}(w-w_0))^{-1/2}dw = \frac{R}{\tau_0 \gamma_s} \leq \frac{R_{\max}}{\tau_{\min} \gamma_s}.$$

Let $\tau_{\max}$ be the maximum distance of flight on the bounded horizon table. Then for each point $r \in \partial \Gamma$ the number of \textquoteleft visible' thermostats cannot exceed the maximal number $S$ of thermostats in any disk of radius $\tau_{\max}$ if we view the positions of thermostats as a periodic configuration on the plane. Therefore $\frac{\cos^2(\varphi)}{\cos^2(\varphi)}$ is less than $\frac{1}{\gamma^2}$ on at most $2S$ intervals and the integral $\int \frac{\cos^2(\varphi')}{\cos^2(\varphi)}dw$ on each is $\leq \frac{2 R_{\max}}{\tau_{\min} \gamma_s}$. Therefore

$$\int\limits^{\pi/2}_{-\pi/2} [\frac{\cos(\varphi') v_\perp}{\cos(\varphi)}]^{-1} \sqrt{\frac{\beta_i}{\pi}} \frac{v_\perp}{\cos^2(\varphi)} e^{- \beta_i v^2_\perp \tan^2(\varphi)} d \varphi $$

$$\leq \gamma_s v_\perp^{-1}+ \sqrt{\frac{\beta_{\max}}{\pi}}  \int\limits_{\frac{\cos(\varphi')}{\cos(\varphi)}<\frac{1}{\gamma_s}} [\frac{\cos^2(\varphi')}{\cos^2(\varphi)}]^{-\frac{1}{2}} \frac{d \varphi}{\cos^2(\varphi)},$$

\begin{equation} \label{eqn:estimate low}
\leq \gamma_s v_\perp^{-1} +
\frac{2 R_{\max}}{\tau_{\min} \gamma_s} \sqrt{\frac{\beta_{\max}}{\pi}} 2S.
\end{equation}

Note that estimate (\ref{eqn:estimate low}) works equally well for all values of $v_\perp$; we will use it later for large $v_\perp$ estimates.

\medbreak

The second integral in the equation (\ref{eqn: 3 part split}) can be bounded by using the fact that $(v_\perp')^2 \leq v_\perp^2 + v_t^2$:
$$\int\limits_{-\frac{\pi}{2}}^{\frac{\pi}{2}} \exp(\epsilon[\frac{\cos(\varphi')v_\perp}{\cos(\varphi)}]^{2}) \rho_{v_\perp}(\varphi) d \varphi = \int\limits_{-\infty}^{\infty} e^{\epsilon (v_\perp')^2} \sqrt{\frac{\beta_i}{\pi}} e^{-\beta_i v_t^2} dv_t$$
\begin{equation} \label{eqn: exp(epsilon v_perp^2) simple bound}
 \leq e^{\epsilon v_\perp^2} \int\limits_{-\infty}^{\infty} \sqrt{\frac{\beta_i}{\pi}} \sqrt{\frac{\beta_i-\epsilon}{\beta_i-\epsilon}} e^{-(\beta_i-\epsilon) v_t^2} dv_t = \sqrt{\frac{\beta_i}{\beta_i-\epsilon}} e^{\epsilon v_\perp^2}
\end{equation}

Thus

$$\int\limits_{-\frac{\pi}{2}}^{\frac{\pi}{2}} V(v_\perp') \rho_{v_\perp}(\varphi) d \varphi $$
$$\leq
\int\limits_{-\frac{\pi}{2}}^{\frac{\pi}{2}} [\frac{\cos(\varphi')v_\perp}{\cos(\varphi)}]^{-1} \rho_{v_\perp}(\varphi) d \varphi + \int\limits_{-\frac{\pi}{2}}^{\frac{\pi}{2}} \exp(\epsilon [\frac{\cos(\varphi')v_\perp}{\cos(\varphi)}]^{2}) \rho_{v_\perp}(\varphi) d \varphi + A\int\limits_{-\frac{\pi}{2}}^{\frac{\pi}{2}} \rho_{v_\perp}(\varphi) d \varphi $$
$$\leq \gamma_s v_\perp^{-1}+ \frac{2 R_{\max}}{ \tau_{\min} \gamma_s} \sqrt{\frac{\beta_{\max}}{\pi}} 2S +\sqrt{\frac{\beta_i}{\beta_i-\epsilon}}\exp[\epsilon (v_\perp^{\min})^2]+A = \gamma_s v_\perp^{-1}+K_s,$$
for any $v_\perp<v_\perp^{\min}$ for some $v_\perp^{\min}$ and $K_s<\infty$.

\subsubsection{Controlling large $v_\perp$} \label{subsect: large v_perp}

To deal with large $v_\perp$ we first note that for any $\varsigma>0$,

$$\int\limits_{-\frac{\pi}{2}}^{\frac{\pi}{2}} V(v_\perp') \rho_{v_\perp}(\varphi) d \varphi$$
$$\leq
\int\limits_{-\frac{\pi}{2}}^{\frac{\pi}{2}} (v_\perp')^{-1} \rho_{v_\perp}(\varphi) d \varphi +  \int\limits_{-\frac{\pi}{2}}^{\frac{\pi}{2}} e^{\epsilon (v_\perp')^2} \rho_{v_\perp}(\varphi) d \varphi + A\int\limits_{-\frac{\pi}{2}}^{\frac{\pi}{2}} \rho_{v_\perp}(\varphi) d \varphi $$
$$\leq \gamma_s v_\perp^{-1} +
\frac{2 R_{\max}}{ \tau_{\min} \gamma_s} \sqrt{\frac{\beta_{\max}}{\pi}} 2S + \int\limits_{-\infty}^{\infty} e^{\epsilon (v_\perp')^2} \sqrt{\frac{
\beta_i}{\pi}} e^{-\beta_i v_t^2} d v_t +A $$
\begin{equation} \label{eqn: varsigma estimate}
\leq \varsigma e^{\epsilon v_\perp}+\int\limits_{-\infty}^{\infty} e^{\epsilon (v_\perp')^2} \sqrt{\frac{
\beta_i}{\pi}} e^{-\beta_i v_t^2} d v_t,
\end{equation}
for any $v_\perp>v_\perp^{\max}$, where $v_\perp^{\max}=v_\perp^{\max}(\varsigma)$ is large enough.

Given $\epsilon>0$ the estimate (\ref{eqn: exp(epsilon v_perp^2) simple bound}) gives a bound by $\sqrt{\frac{\beta_i}{\beta_i-\epsilon}} e^{\epsilon v_\perp^{\max}}$ when $v_\perp \leq v_\perp^{\max}$. However, this estimate is too rough to produce uniform $\gamma_l<1$ for large $v_\perp \geq v_\perp^{\max}$, so we need to be more careful here.

For any $\delta>0$, and $\Lambda$ large enough, using Chebyshev's inequality, we have
$$\int\limits_{|v_t| > \Lambda} e^{\epsilon (v_\perp')^2} \sqrt{\frac{\beta_i}{\pi}} e^{-\beta_i v_t^2} dv_t$$
$$\leq e^{\epsilon v_\perp^2} \sqrt{\frac{\beta_i}{\beta_i-\epsilon}} \int\limits_{|v_t| > \Lambda} \sqrt{\frac{\beta_i-\epsilon}{\pi}} e^{-(\beta_i-\epsilon) v_t^2} dv_t  \leq e^{\epsilon v_\perp^2} \sqrt{\frac{\beta_i}{\beta_i-\epsilon}} \frac{1}{2(\beta_i-\epsilon)\Lambda^2}< \frac{\delta}{2} e^{\epsilon v_\perp^2}.$$
Note that $\Lambda$ can be chosen independent of $\epsilon$ provided $\epsilon < \epsilon_0 < \min\{\beta_1, \cdots, \beta_p\}$.

In the remaining the integral
$$\int\limits_{|v_t| \leq \Lambda} e^{\epsilon (v_\perp')^2} \sqrt{\frac{\beta_i}{\pi}} e^{-\beta_i v_t^2} dv_t=\int\limits_{|
\tan(\varphi)|=|\frac{v_t}{v_\perp}| \leq \frac{\Lambda}{v_\perp}} e^{\epsilon (v_\perp')^2} \rho_{v_\perp}(\varphi)d\varphi,$$
when $v_\perp \geq v_\perp^{\max}$ with $v_\perp^{\max}$ large enough, we only integrate over values of $\varphi \approx 0$. So our new perpendicular velocity $v_\perp' = \frac{\cos(\varphi')}{\cos(\varphi)}v_\perp \approx \cos(\varphi_0')v_\perp$ depends for the most part on the landing angle $\varphi_0'$ which corresponds to $\varphi=0$. For most locations $r \in \partial \Gamma$, $\cos(\varphi_0')$ is uniformly bounded away from $1$.

Before proceeding with further estimates, we are going to identify troublesome geometric locations. For this purpose we are going to lift the periodic configuration of thermostats from a torus to a plane.

Let $I$ be a line segment connecting the centers of two thermostats $D_i$ and $D_j$ such that $I$ does not intersect with any other thermostats. Then if a particle originates at $r = I \cap \partial D_i$ in the normal direction, i.e. with $\varphi=0$, it must hit $D_j$ at $r'$ also in the normal direction, i.e. with $\varphi'=0$. In addition, the normal velocity does not change along such a trajectory, i.e. $v_\perp'=v_\perp$. Suppose now that initial $v_\perp$ is extremely large. Then the tangential velocity $v_t$ acquired before the flight is, on average, tiny compared to $v_\perp$ and the particle lands very close to $r'$ with $\varphi' \approx 0$; i.e. $v_\perp' \approx v_\perp$ and we do not get much of a drop in $v_\perp$. And it is clear that the higher the initial velocity $v_\perp$, the smaller is the drop. This indicates that locations like $r$ and $r'$ are \textquoteleft troublesome' and need to be treated with extra care.

Given a thermostat $D_i$, consider all line segments $I_j$ on the lift to the plane connecting its center to the centers of other thermostats, possibly intersecting other thermostats, and of length at most $2R_{\max}+\tau_{\max}$, where $R_{\max}=\max\{R_1, \cdots, R_p\}$ and $\tau_{\max}$ is the maximal length of collision-free path. Note that there are only finitely many such segments $I_j$ since there is an upper bound on the length. Let $r_j=I_j \cap \partial D_i$. If a particle originates from $r \in \partial D_i$ in the normal direction denote the point on next collision with a thermostat by $r_0'$ (tangential collisions count as collisions).

Let $\Delta r_i$ be small enough such that for any $j$ the intervals $[r_j - \Delta r_i, r_j + \Delta r_i]$ do not intersect and for all $r \in [r_j - \Delta r_i, r_j + \Delta r_i]$, all $r_0'$ belong to the same thermostat unless $I_j$ touches a thermostat tangentially. In the latter case we require the same for all $r \in [r_j - \Delta r_i, r_j )$ and for all $r \in (r_j, r_j + \Delta r_i]$.

Let $\tilde{\Delta} r < \min_i\{\frac{R_{\min}\Delta r_i}{2}\}$ and denote by $\mathcal{R}$ the union of all $r_j$'s over all thermostats $D_i$. In addition, consider all segments $I$ connecting the centers of thermostats but not intersecting any other thermostats except tangentially. Denote the collection of possible intersections of such $I$ with $\partial \Gamma$ by $\{r_k\}$. Note that $\{r_k\} \subset \mathcal{R}$ and that if
$r \in (r_k-\tilde{\Delta} r,r_k+\tilde{\Delta} r)$, for all $\varphi$, such that $|\varphi| \leq \frac{r}{R_{\max}}$, $r'$ belongs to the same thermostat as $r_0'$ (since $(r_k+2(r-r_k))_0'$ does).

\begin{lemma} \label{lemma:angle bounds}
Suppose $1$-segment path originates from $(r,v_\perp)$. Then we can choose $\Delta r \leq \tilde{\Delta} r$  small enough such that
\begin{itemize}
\item if $r \not \in (r_k-\Delta r,r_k+\Delta r)$ for all $k$ and $|\varphi| \leq \frac{\Delta r}{R_{\max}}$, then $|\varphi'| \geq \frac{\Delta r}{2 R_{\max}}$.
\item if $r \in (r_k-\Delta r,r_k+\Delta r)$ for some $k$ and $|\varphi| \leq \frac{|r-r_k|}{R_{\max}}$, then $|\varphi'| \geq \frac{|r-r_k|}{2 R_{\max}}$
\end{itemize}
\end{lemma}

\medbreak
\begin{proof}
Choose $\Delta r \leq \tilde{\Delta} r$  such that $\sin(\frac{\Delta r}{R_{\min}}) \geq \frac{\Delta r}{2 R_{\min}}$. We are going to focus on the first statement. The proof of the second statement is essentially the same noting that the situation with tangential collision would only give a better estimate.

Denote the line segment from $(r,\varphi)$ to $(r', \varphi')$ by $\gamma$, the thermostat containing $r$ by $D$, the thermostat containing $r'$ by $D'$, radius of $D$ by $R$, and radius of $D'$ by $R'$. Let $I$ be the line segment connecting the center $O$ of the thermostat $D$ to the center $O'$ of the thermostat $D'$ and let $r_I=I \cap \partial D$ and $r_I'=I \cap \partial D'$. Note that $|I| \leq 2R_{\max}+\tau_{\max}$ since $|\gamma| \leq \tau_{\max}$. By assumptions on $\Delta r$, we conclude that $|r-r_I| \geq \Delta r$.

Let $J$ be the line segment parallel to $I$ originating at $r$ and ending at some $r_{\parallel}'$. Denote the angle $J$ forms with the normal to $\partial D$ by $\varphi_0$ and with the normal to the disk of next collision by $\varphi_0'$. Then $|\varphi_0| = \frac{|r-r_I|}{R} \geq \frac{\Delta r}{R} \geq |\varphi|$. Therefore $r_{\parallel}' \in \partial D'$ and $|\varphi'| \geq |\varphi_0'|$.

Then $|\varphi'| \geq |\varphi_0'|=|\angle r_I' O' r_{\parallel}'| \geq |\sin(\angle r_I' O' r_{\parallel}')| =|\frac{R \sin(\angle r_I 0 r)}{R'}| \geq \frac{R \Delta r}{R' 2R} \geq \frac{\Delta r}{2 R_{\max}}$.
\end{proof}

\medbreak

If $r \not \in (r_k-\Delta r,r_k+\Delta r)$, then for $v_t \leq \Lambda$ and $v_\perp \geq v_\perp^{\max} \geq \frac{\Lambda R_{\max}}{\Delta r}$ large enough, $\varphi \leq \tan(\varphi) \leq \frac{\Lambda}{v_\perp} \leq \frac{\Delta r}{R_{\max}}$ and
$$(v_\perp')^2=(v_\perp^2+v_t^2)\cos^2(\varphi') \leq v_\perp^2 - \sin^2(\frac{\Delta r}{2 R_{\max}})v_\perp^2+\Lambda^2 \leq v_\perp^2 - \alpha v_\perp^2,$$
where $\alpha = \frac{1}{2}\sin^2(\frac{\Delta r}{2 R_{\max}})$.

Given $\delta>0$ and $\Lambda$ satisfying $\sqrt{\frac{\beta_i}{\beta_i-\epsilon}} \frac{1}{2(\beta_i-\epsilon)\Lambda^2}< \frac{\delta}{2}$ we get
$$\int\limits_{-\infty}^{\infty} e^{\epsilon (v_\perp')^2} \sqrt{\frac{\beta_i}{\pi}} e^{-\beta_i v_t^2} dv_t$$
$$\leq \int\limits_{|v_t| > \Lambda} e^{\epsilon (v_\perp')^2} \sqrt{\frac{\beta_i}{\pi}} e^{-\beta_i v_t^2} dv_t + \int\limits_{|v_t| \leq \Lambda} e^{\epsilon (v_\perp')^2} \sqrt{\frac{\beta_i}{\pi}} e^{-\beta_i v_t^2} dv_t$$
$$\leq \frac{\delta}{2} e^{\epsilon v_\perp^2} + e^{-\epsilon \alpha v_\perp^2}e^{\epsilon v_\perp^2} \leq (\frac{\delta}{2}+1-\delta)e^{\epsilon v_\perp^2} = (1-\frac{\delta}{2})e^{\epsilon v_\perp^2}$$
provided that $v_\perp$ is large enough so that $e^{-\epsilon \alpha v_\perp^2} \leq 1-\delta$.

Therefore we need to focus our attention on studying the dynamics in regions $(r_k-\Delta r,r_k+\Delta r)$. Let $r \in (r_k-\Delta r,r_k+\Delta r)$ and $r'$ be the image of $r$ given the tangential velocity $v_t \leq \Lambda$. First let us note that we can ignore the situation when intervals $I$ connecting centers of thermostats intersect some other thermostats tangentially since it would only give us a better estimate. Then if $v_\perp \geq v_\perp^{\max} \geq \frac{\Lambda R_{\max}}{\Delta r}$ all appropriate images $r'$ lie on the thermostat $D_k'$.

Let
$$y=\frac{r v_\perp}{R_k} \;\;\; and \;\;\; y'=\frac{r' v_\perp}{R_k'}$$

\begin{lemma} \label{lemma:drop by alpha}
Given any $\Lambda>0$ and $A>0$, there exist $Y=Y(A,\Lambda)$ and $v_\perp^{\max}=v_\perp^{\max}(A,\Lambda)$ such that
$$(v_\perp')^2 - v_\perp^2 \leq -A$$ provided that $|y| \geq Y$, $|v_\perp| \geq v_\perp^{\max}$ and $|v_t| \leq \Lambda$.
\end{lemma}

\begin{proof} Pick $Y = \max\{R_{\max} \Lambda, 8 R_{\max}(A+\Lambda^2) \}$. Then if $|y| \geq Y$, $|r-r_k| \geq
\frac{Y}{v_\perp}$. And by Lemma  \ref{lemma:angle bounds}, since $|\varphi| \leq |\tan(\varphi)| \leq \frac{\Lambda}{v_\perp} \leq \frac{Y}{R_{\max} v_\perp}$,
$$(v_\perp')^2 - v_\perp^2 \leq -\sin^2(\frac{Y}{2 R_{\max} v_\perp})v_\perp^2+ \Lambda^2 \leq -\frac{Y^2}{8 R_{\max}^2} + \Lambda^2 \leq - A$$
provided that $v_\perp \geq v_\perp^{\max}$ is large enough.
\end{proof}

Then if we choose $A$ such that $e^{-\epsilon A} \leq 1-\delta$, for $|y| \geq Y$ and $|v_\perp| \geq v_\perp^{\max}$,
$$\int\limits_{-\infty}^{\infty} e^{\epsilon (v_\perp')^2} \sqrt{\frac{\beta_i}{\pi}} e^{-\beta_i v_t^2} dv_t \leq \frac{\delta}{2} e^{\epsilon v_\perp^2} + e^{-\epsilon A}e^{\epsilon v_\perp^2} \leq (\frac{\delta}{2}+1-\delta)e^{\epsilon v_\perp^2} = (1-\frac{\delta}{2})e^{\epsilon v_\perp^2}$$

It remains to treat $|y| \leq Y$ case.

\begin{lemma} \label{lemma:uniform convergence}
Suppose $|v_t| \leq \Lambda$ and $|y| \leq Y$. Then
$$(v_\perp')^2 = v_\perp^2 + [\frac{d^2}{(R_k')^2}+\frac{2d}{R_k'}]v_t^2 +[R_k+R_k'+d]^2\frac{y^2}{(R_k')^2}$$
$$+ [2dR_k' + R_kd+d^2+(R_k')^2+R_kR_k']\frac{2 v_t y}{(R_k')^2}+\mathcal{E}(v_\perp),$$
where $\mathcal{E}(v_\perp) \to 0$ as $v_\perp \to \infty$ uniformly in $\{(y,v_t): |y|\leq Y, |v_t| \leq \Lambda\}$.
Here $R_k$ and $R_k'$ are radii of $D_k$ and $D_k'$ respectively and $d=d_{k,k'}$ is the distance between $D_k$ and $D_k'$.
\end{lemma}

\begin{proof}
From geometry
$$r' = \frac{R_k+d}{R_k}r + \frac{d} {v_\perp}v_t + O(r^2) \;\;\; as \; r \to 0$$
Then, since $|y|\leq Y$
$$y'=\frac{R_k+d}{R_k'}y + \frac{d}{R_k'}v_t + O(\frac{1}{v_\perp}) \;\;\; as \; v_\perp \to \infty$$
Note that both $O(r^2)$ and $O(\frac{1}{v_\perp})$ are uniform in $\{(y,v_t): |y|\leq Y, |v_t| \leq \Lambda\}$.

Let us orient the appropriate line segment $I$ connecting centers of $D_k$ and $D_k'$ vertically and let $\theta=\frac{r}{R_k}$ and $\theta'=\frac{r'}{R_k'}$ be the angles $r$ and $r'$ forms with $I$. Then
$$v_{vert}=v_\perp \cos(\theta) - v_t \sin(\theta)$$
$$v_{hor}=v_\perp \sin(\theta)+ v_t \cos(\theta)$$
and
$$v_\perp'=v_{vert} \cos(\theta)-v_{hor}\sin(\theta)$$
$$= v_\perp \cos(\theta)\cos(\theta')-v_t
\sin(\theta) \cos(\theta')-v_\perp \sin(\theta)\sin(\theta')-v_t\cos(\theta)\sin(\theta')$$
$$=v_\perp(1-\frac{1}{2}\frac{y^2}{v_\perp^2}-\frac{1}{2}\frac{(y')^2}{v_\perp^2})-v_t \frac{y}{v_\perp}1-v_\perp(\frac{y}{v_\perp}\frac{y'}{v_\perp})-v_t\frac{y'}{v_\perp}1 + O(\frac{1}{v_\perp^2})$$
$$=v_\perp + \frac{1}{2}\frac{(y^2+(y')^2+2v_ty +2yy'+2v_ty')}{v_\perp} + O(\frac{1}{v_\perp^2}).$$

So
$$(v_\perp')^2-v_\perp^2=(y^2+(y')^2+2v_ty+2yy'+2v_ty')+O(\frac{1}{v_\perp^2})$$
$$=[y+\frac{R_k+d}{R_k'}y+\frac{d}{R_k'}v_t]^2+ 2v_t[y+\frac{R_k+d}{R_k'}y+\frac{d}{R_k'}v_t]+\mathcal{E}(v_\perp)=[\frac{d^2}{(R_k')^2}+\frac{2d}{R_k'}]v_t^2$$
$$+\frac{2dv_t}{R_k'}y+\frac{2(R_k+d)d v_t}{(R_k')^2}y+2v_ty+2v_ty\frac{R_k+d}{R_k'}+[R_k+R_k'+d]^2\frac{y^2}{(R_k')^2}+\mathcal{E}(v_\perp)$$
$$=[\frac{d^2}{(R_k')^2}+\frac{2d}{R_k'}]v_t^2 + [2dR_k' v_t + 2(R_k+d)dv_t+2v_t (R_k')^2+2v_t(R_k+d)R_k']\frac{y}{(R_k')^2}$$
$$+ [R_k+R_k'+d]^2\frac{y^2}{(R_k')^2}+\mathcal{E}(v_\perp)=[\frac{d^2}{(R_k')^2}+\frac{2d}{R_k'}]v_t^2 +[R_k+R_k'+d]^2\frac{y^2}{(R_k')^2}$$
$$+[2dR_k' + R_kd+d^2+(R_k')^2+R_kR_k']\frac{2 v_t y}{(R_k')^2} +\mathcal{E}(v_\perp), $$

where $\mathcal{E}(v_\perp) = O(v_\perp)$ as $v_\perp \to \infty$ uniformly in $\{(y,v_t): |y|\leq Y, |v_t| \leq \Lambda\}.$
\end{proof}

\medbreak

We are now ready to estimate $e^{-\epsilon v_\perp}\int\limits_{|v_t| \leq \Lambda} e^{\epsilon (v_\perp')^2} \sqrt{\frac{\beta_i}{\pi}} e^{-\beta_i v_t^2} dv_t$ for $y \leq Y$ and $v_\perp \geq v_\perp^{\max}$.
$$\sqrt{\frac{\beta_i}{\pi}} \exp(-\beta_i v_t^2+\epsilon (v_\perp')^2-\epsilon v_\perp^2)
=\sqrt{\frac{\beta_i}{\pi}}\exp(-[\beta_i+\epsilon \frac{d^2}{(R_k')^2}+\epsilon \frac{2d}{R_k'}]v_t^2$$
$$-[(R_k'+d)^2+dR_k+R_kR_k']\frac{2yv_t}{(R_k')^2}\epsilon-[R_k+R_k'+d]^2\frac{y^2}{(R_k')^2}\epsilon+\mathcal{E}(v_\perp))$$
$$=\sqrt{\frac{\beta_i}{\beta_i'}} \sqrt{\frac{\beta_i'}{\pi}} \exp(-\beta_i'(v_t+\frac{y\tilde{R}\epsilon}{(R_k')^2 \beta_i'})^2) $$
$$\times \exp(-\frac{y^2}{(R_k')^2}[(R_k+R_k'+d)^2\epsilon+\mathcal{E}(v_\perp)-\frac{\tilde{R}^2 }{(R_k')^2\beta_i'}\epsilon^2]),$$
where $\beta_i'=\beta_i'(\epsilon)=\beta_i+\epsilon \frac{d^2}{(R_k')^2}+\epsilon \frac{2d}{R_k'}$ and $\tilde{R}=(R_k'+d)^2+dR_k+R_kR_k'$.

\medbreak

Let $\epsilon \leq \epsilon_0$ be small enough and $\tilde{v}_\perp^{\max}$ large enough such that $$[R_k+R_k'+d]^2\epsilon+\mathcal{E}(v_\perp)-\frac{\tilde{R}^2 }{(R_k')^2\beta_i'}\epsilon^2 \geq 0$$
for $v_\perp \geq \tilde{v}_\perp^{\max}$ (can do this by Lemma \ref{lemma:uniform convergence}). Then
$$\int\limits_{|v_t|\leq \Lambda} \sqrt{\frac{\beta_i}{\pi}} e^{-\beta_i v_t^2+\epsilon (v_\perp')^2-\epsilon v_\perp^2} \leq \sqrt{\frac{\beta_i}{\beta_i'}} \int\limits_{-\infty}^{\infty} \sqrt{\frac{\beta_i'}{\pi}} e^{-\beta_i'(v_t+\frac{y\tilde{R}\epsilon}{(R_k')^2 \beta_i'})^2} = \sqrt{\frac{\beta_i}{\beta_i'}}$$

Let
$$\gamma=\max\limits_{\{R_k,R_k',d_{k,k'},\beta_i\}} \sqrt{\frac{\beta_i}{\beta_i'}}=\max\limits_{\{R_k,R_k',d_{k,k'},\beta_i\}} \sqrt{\frac{\beta_i}{\beta_i+\epsilon \frac{d^2}{(R_k')^2}+\epsilon \frac{2d}{R_k'}}}.$$
Note that $\gamma < 1$.
Choose $\delta$ and $\varsigma$ so small that $\gamma + \delta + \varsigma <1$ and $1-\frac{\delta}{2}+\varsigma < 1$ and then choose $\Lambda$ so large that $\sqrt{\frac{\beta_{\min}}{\beta_{\min}-\epsilon}}\frac{1}{2(\beta_{\min}-\epsilon)\Lambda^2}< \frac{\delta}{2}$.
Let $\gamma_l = \max\{\gamma + \delta + \varsigma, 1-\frac{\delta}{2}+\varsigma\}$. Choose $v_\perp^{\max} \geq \tilde{v}_\perp^{\max}$ that satisfies Lemma \ref{lemma:drop by alpha} and all the other assumptions we made.
Then for $v_\perp \geq v_\perp^{\max}$
$$\int\limits_{-\frac{\pi}{2}}^{\frac{\pi}{2}} V(v_\perp') \rho_{v_\perp}(\varphi) d \varphi \leq \gamma_l \exp^{\epsilon v_\perp}=\gamma_l V(v_\perp).$$

\medbreak

Let $\gamma=\max \{\gamma_s,\gamma_l \}$ and $K=\max\{K_s,\sqrt{\frac{\beta_{\min}}{\beta_{\min}-\epsilon}}e^{\epsilon (v_\perp^{\max})^2}\}$. Then

$$ \mathcal{P}^*V(v_\perp)= \int\limits^{\pi/2}_{-\pi/2} V(v'_\perp) \rho_{v_\perp}(\varphi) d \varphi \leq \gamma V(v_\perp) + K, \; \; \; \forall (r,v_\perp) \in \Omega. \;\; \; $$

\medbreak

\section{Minorization Condition} \label{sect:minorization}

For any $v_\perp^{\min} < v_\perp^{\max}$, let $\mathcal{C}=\{(r,v_\perp): v_\perp^{\min} \leq v_\perp \leq v_\perp^{\max}\}$ and let $\nu$ be uniform probability measure on $\mathcal{C}$.

\begin{proposition} \label{prop: minorization condition}
There exist an integer $N>0$ and $\eta_N>0$ such that
$$\inf\limits_{(r,v_\perp) \in \mathcal{C}}\mathcal{P}^{N}((r,v_\perp),\cdot) \geq \eta_N\nu(\cdot) $$
In addition, there exists $\eta_{N+1}>0$ such that the same holds for $(N+1)^{st}$ push forward.
\end{proposition}

We are going to prove Prop. \ref{prop: minorization condition} in three steps: (i) Prop \ref{prop:sample path}, which guarantees existence of a \textquoteleft regular' $N$-step sample path from any $(r,v_\perp) \in \tilde{\mathcal{C}}$ to any $(r',v_\perp') \in \tilde{\mathcal{C}}$, where $\tilde{\mathcal{C}}=\{(r,v_\perp): \tilde{v}_\perp^{\min} \leq v_\perp \leq \tilde{v}_\perp^{\max}\}$ for some $\tilde{v}_\perp^{\min} < v_\perp^{\min}$ and $\tilde{v}_\perp^{\max}>v_\perp^{\max}$; (ii) Prop. \ref{prop: pushing density forward}, which gives bounds on the densities when pushed forward along regular paths; and (iii) Prop. \ref{prop: acquiring density}, which guarantees that for any $(r,v_\perp) \in \mathcal{C}$ we can \textquoteleft acquire' density in two steps and the resulting location $(r'',v_\perp'') \in \tilde{\mathcal{C}}$. So, originating at $(r,v_\perp) \in \mathcal{C}$, we acquire density around $(r'',v_\perp'') \in \tilde{\mathcal{C}}$ by Prop \ref{prop: acquiring density} and then push it forward along the sample path from Prop \ref{prop:sample path}; and Prop. \ref{prop: pushing density forward} ensures that we get a uniform lower bound on the density at the final point $(r',v_\perp')$

To state Prop. \ref{prop:sample path},  \ref{prop: pushing density forward}, and \ref{prop: acquiring density}, we need the following definitions.

\begin{definition}
A \emph{projected (particle) path} is a continuous curve $\gamma:[0,1] \to \Gamma$, $s \mapsto \gamma(s)$ that
consists of a finite sequence of straight segments meeting at $\partial \Gamma$ with $\gamma(0)$, $\gamma(1) \in \partial \Gamma$.
\end{definition}

\begin{remark}
Note that a projected path is allowed to have \emph{any} \textquoteleft reflections' off the boundaries of the thermostats $\cup_{j=1}^{N} \partial D_j$. An example of a projected path is shown in Fig \ref{fig: configuration}.
\end{remark}

If we start a projected path with certain normal velocity $v_\perp$, then all subsequent normal velocities upon collisions are completely determined by the path. More precisely, at $k^{th}$ collision
$$v_\perp^k = \frac{\cos(\varphi_1') \cdots \cos(\varphi_k')}{\cos(\varphi_1) \cdots \cos(\varphi_k)}v_\perp,$$
where $\varphi_1, \cdots, \varphi_k$ and $\varphi_1', \cdots, \varphi_k'$ are the angles of reflections and incidences for the path. Note also that multiplying the initial $v_\perp$ by a certain factor multiplies all subsequent $v_\perp^k$ by the same factor.

\begin{definition}
Let $\gamma$ be a projected path that starts at $r$ and ends at $r'$. A sample path $\sigma$ given $\gamma$ and $v_\perp$ is a path in the phase space that starts at $(r, v_\perp)$ with positions updated according to positions of the path $\gamma$ and normal velocities updated as above.
\end{definition}

Note that given $\gamma$ and $v_\perp$, $\sigma$ is unique and to each sample path $\sigma$ corresponds a unique projected path.

\begin{definition}
Given $\delta>0$, $\Delta r>0$, and $\Delta \varphi>0$ a projected path $\gamma$ is $(\delta,\Delta r, \Delta \varphi)$-regular if
\begin{itemize}
\item All the incoming and outgoing angles are bounded away from $\pm\frac{\pi}{2}$ by $\delta$, i.e. all $\varphi, \varphi' \in (-\frac{\pi}{2}+\delta, \frac{\pi}{2}-\delta)$.
\item For each segment of $\gamma$ from $(r_k,\varphi_k)$ to $(r_{k+1},\varphi_k')$, any segment originating from $\tilde{r}_k \in (r_k-\Delta r, r_k+\Delta r)$ with angle $\tilde{\varphi}_k \in (\varphi_k - \Delta \varphi, \varphi_k+\Delta \varphi)$ has incoming position $\tilde{r}_{k+1}$ on the same thermostat as $r_{k+1}$ and angle $\tilde{\varphi}_k' \in (-\frac{\pi}{2}+ \delta, \frac{\pi}{2} - \delta)$; in addition, the same holds for each segment in reverse direction, i.e. from $(r_{k+1},\varphi_k')$ to $(r_k,\varphi_k)$.
\end{itemize}
\end{definition}

\begin{proposition} \label{prop:sample path}
There exist $N$, $\delta_0$, $\Delta r_0$, and $\Delta \varphi_0$ such that given any $(r,v_\perp)$ and $(r',v_\perp')$ with $\tilde{v}_\perp^{\min} \leq v_\perp,v_\perp' \leq \tilde{v}_\perp^{\max}$, i.e. $(r,v_\perp), (r',v_\perp') \in \tilde{\mathcal{C}}$, there exists a sample path from $(r,v_\perp)$ to $(r',v_\perp')$ such that the corresponding projected path
\begin{enumerate}
\item is $(\delta_0, \Delta r_0, \Delta \varphi_0)$-regular.
\item has exactly $N$ collisions.
\end{enumerate}
\end{proposition}

Here $\tilde{v}_\perp^{\min} = \cos^2(\frac{\pi}{2}-\delta_0) v_\perp^{\min}$ and $\tilde{v}_\perp^{\max} = \frac{v_\perp^{\max}}{\cos^2(\frac{\pi}{2}-\delta_0)}$.

We are going to prove Prop. \ref{prop:sample path} in three steps. The first two are covered by the following Lemmas.

\begin{lemma}\label{lemma:path}
There exist $K$, $\delta_0$, $\Delta r_0$, and $\Delta \varphi_0$ such that for any choice of $r,r'$ there exists a projected $(\delta_0, \Delta r_0, \Delta \varphi_0)$-regular path from $r$ to $r'$ in $\Gamma$ with $\leq K$ collisions.
\end{lemma}

\begin{lemma} \label{lemma:v_perp boost}
There exist $M$, $\delta_0$, $\Delta r_0$, and $\Delta \varphi_0$ such that given any $(r,v_\perp), (r',v_\perp') \in \tilde{\mathcal{C}}$ there exists a sample path from $(r,v_\perp)$ to $(r',v_\perp')$ with $\leq M$ collisions and the corresponding projected path is $(\delta_0, \Delta r_0, \Delta \varphi_0)$-regular.
\end{lemma}

The next two propositions deal with densities and are used in the proof of Prop. \ref{prop: minorization condition}.

\begin{proposition} \label{prop: pushing density forward}
Given $\Delta r<\Delta r_0$, $\Delta v_\perp$ small enough, there exists $\tilde{\eta}_N$ such that if $\mu$ is the uniform probability measure in the neighborhood $(r-\Delta r,r+ \Delta r) \times (v_\perp-\Delta v_\perp,v_\perp+ \Delta v_\perp)$ of $(r,v_\perp) \in \tilde{\mathcal{C}}$, then $\mathcal{P}^N_* \mu \geq \tilde{\eta}_N \tilde{\nu}$, where $\tilde{\nu}$ is the uniform probability measure on $\tilde{C}$.
\end{proposition}

\begin{proposition} \label{prop: acquiring density}
There exist $\eta$, $\Delta r$, and $\Delta v_\perp$ such that for any $(r,v_\perp) \in \mathcal{C}$, $\mathcal{P}^2_* \delta_{(r,v_\perp)} \geq \eta \nu_{(r,v_\perp)}$, where $\nu_{(r,v_\perp)}$ is the uniform probability measure on $A_{(r,v_\perp)}=(r''-\Delta r,r''+\Delta r ) \times (v_\perp''-\Delta v_\perp, v_\perp''+\Delta v_\perp)$ for some $(r'',v_\perp'') \in \tilde{\mathcal{C}}$.
 \end{proposition}

\begin{remark}
In order to satisfy the conditions of Prop. \ref{prop: minorization condition} for both $N^{th}$ and $(N+1)^{st}$ push forward of $\Phi$, we first note that Prop. \ref{prop:sample path} holds similarly if $(r,v_\perp), (r',v_\perp') \in \hat{\mathcal{C}}$, where $\hat{v}_\perp^{\min} = \cos^3(\frac{\pi}{2}-\delta_0) v_\perp^{\min}$ and $\hat{v}_\perp^{\max} = \frac{v_\perp^{\max}}{\cos^3(\frac{\pi}{2}-\delta_0)}$ since $v_\perp^{\min}$ and $v_\perp^{\max}$ are arbitrary in the first place. Therefore to produce an $N+1$ segment path from $(r,v_\perp) \in \mathcal{C}$ to $(r',v_\perp') \in \mathcal{C}$, one takes a one step regular sample path from $(r,v_\perp) \in \mathcal{C}$ to some $(\hat{r},\hat{v}_\perp) \in \hat{\mathcal{C}}$ and then follow an $N$ step regular sample path to $(r',v_\perp')$. Prop. \ref{prop: pushing density forward} easily applies for $N+1$ step regular paths. See proof in subsect. \ref{subsect:push forward}.
\end{remark}

\subsection{Proof of Lemma \ref{lemma:path}}
Our periodic bounded horizon billiard configuration with the maximal flight distance denoted by $\tau_{\max}$ has the following geometric property:

\begin{lemma}\label{lemma: angle property}
There exist $\delta_1$, $\Delta r_1$, and $\Delta \varphi_1$ such that for each $r \in \cup_i \partial D_i$ there exists a $1$-segment projected $(\delta_1, \Delta r_1, \Delta \varphi_1)$-regular path $\gamma$ that originates at $r$.
\end{lemma}

Lemma \ref{lemma: angle property} is a rather simple consequence of continuity and compactness. We include a proof for completeness purposes at the end of this subsection.

\begin{definition}
Given $\delta_1$, $\Delta r_1$, and $\Delta \varphi_1$ from Lemma \ref{lemma: angle property}, for any $r \in \partial D_i$ let $U_r$ be set of points in $\partial D_i$ that can be reached from $r$ by a $2$-segment projected $(\delta_1, \Delta r_1, \Delta \varphi_1)$-regular path.
\end{definition}

Each $U_r$ is non-empty since $r \in U_r$ by Lemma \ref{lemma: angle property} and open by definition of a projected $(\delta_1, \Delta r_1, \Delta \varphi_1)$-regular path. Let $J_r$ be the connected component of $U_r$ that contains $r$. Then $\cup_r J_r$ is a covering of a compact set $\partial D_i$. Let $\mathcal{J}=\{J_1, \cdots, J_{n_i}\}$ be a finite sub-covering such that no element of $\mathcal{J}$ is fully contained in the union of some other elements of $\mathcal{J}$. The later minimal property guarantees that each point in $\partial D_i$ is covered by at most $2$ elements of $\mathcal{J}$. Let $r_1, \cdots, r_{n_i}$ be the set of points that generated the original $J_i := J_{r_j}$ in the subcovering and $\tilde{r}_1, \cdots, \tilde{r}_{n_i}$ be the set of midpoints of all the non-empty intersections of elements of $\mathcal{J}$.

Further re-order $J_{1}, \cdots, J_{n_i}$ and $\tilde{r}_1, \cdots, \tilde{r}_{n_i}$ clockwise such that $\tilde{r}_k \in J_{k} \cap J_{k+1}$. Then there exists a projected $(\delta_1, \Delta r_1, \Delta \varphi_1)$-regular path $\gamma_i: r_1 \to s_1 \to \tilde{r}_1 \to \tilde{s}_1 \to r_2 \cdots \to r_{n_i} \to s_{n_i} \to \tilde{r}_{n_i} \to \tilde{s}_{n_i} \to r_1$ around the thermostat with $s_1, \tilde{s}_1, \cdots, s_{n_i}, \tilde{s}_{n_i}$ chosen to satisfy the $(\delta_1, \Delta r_1, \Delta \varphi_1)$-regularity property.

Assume now that $r, r'$ are any two points on $\partial D_i$. Then for some $l$ and $m$, $r, r_l \in J_l$ and $r',r_m \in J_m$ and there exist projected $(\delta_1, \Delta r_1, \Delta \varphi_1)$-regular paths $r \to s \to r_l$ and $r_m \to s' \to r'$. Therefore we can construct a projected $(\delta_1, \Delta r_1, \Delta \varphi_1)$-regular path from $r$ to $r'$ that starts with $r \to s \to r_l$, follows $\gamma_i$ from $r_l$ to $r_m$, and ends with $r_m \to s' \to r'$. The number of total collisions for such a path is bounded by $K_i=4n_i+4$.

For any two thermostats there is a finite chain of length at most $\leq p$ of thermostats joining these two such that any two neighbors in this chain can be connected by some $1$-segment projected paths that do not meet thermostats tangentially. Let $\{\tilde{\gamma_j}\}$ be a finite collection of such $1$-segment projected paths such that the above connectivity property is satisfied using paths from $\{\tilde{\gamma_j}\}$ only. Let $\delta_2 = \frac{\pi}{2} - \max_\varphi\{|\varphi|\}$, where the maximum is taken over all incoming and outgoing angles $\varphi$ of all $\tilde{\gamma}_j$. And let $\Delta r_2$ and $\Delta \varphi_2$ be such that all $\tilde{\gamma}_j$ are $(\delta_2, \Delta r_2, \Delta \varphi_2)$-regular.

Then if we choose $\delta_0=\min \{\delta_1, \delta_2\}$, $\Delta r_0=\min \{\Delta r_1, \Delta r_2\}$, $\Delta \varphi_0=\min \{\Delta \varphi_1, \Delta \varphi_2\}$, and $K= \sum_{1 \leq i \leq p}\{K_i\} + p$, where $p$ is the number of thermostats in the system. Then there exists a projected $(\delta_0, \Delta r_0, \Delta \varphi_0)$-regular path from any $r \in \cup_i D_i$ to any $r' \cup_i D_i$ with the number of collisions bounded by $K$.

\vskip 0.3in

\begin{proof}[of Lemma \ref{lemma: angle property}]
Assume for $\delta_n=\Delta r^{(n)} = \Delta \varphi^{(n)}=\frac{1}{n}$, there exists a sequence $r_n$ such that no $1$-segment projected path originating at $r_n$ is $(\frac{1}{n},\frac{1}{n},\frac{1}{n})$-regular. Let $r$ be an accumulation point of $r_n$. Then by continuity there exists a $1$-segment projected path originating from $r$ that is $(\frac{1}{m},\frac{1}{m},\frac{1}{m})$-regular for some $m$. Since $(\frac{1}{m},\frac{1}{m},\frac{1}{m})$-regularity is an open condition, for each point $\tilde{r}$ in a small neighborhood around $r$, there exists a $1$-segment projected  $(\frac{1}{m},\frac{1}{m},\frac{1}{m})$-regular path originating from $\tilde{r}$, a contradiction.
\end{proof}

\subsection{Proof of Lemma \ref{lemma:v_perp boost}}
Suppose we start with $r$ and move to $r'$ along the projected $(\delta_0, \Delta r_0, \Delta \varphi_0)$-regular path from Lemma \ref{lemma:path}. Given $v_\perp$, let $\sigma_0$ be the corresponding sample path that starts at $(r,v_\perp)$ and ends at $(r',\tilde{v}_\perp)$ for some $\tilde{v}_\perp$. We would like to extend $\sigma_0$ by $\sigma_1$ such that $\sigma_1$ starts at $(r',\tilde{v}_\perp)$ and ends at $(r', v_\perp')$.

First of all, since the path in Lemma \ref{lemma:path}  is $(\delta_0, \Delta r_0, \Delta \varphi_0)$-regular with $\leq K$ segments $\tilde{v}_\perp$ has the following bounds:
$$\cos^K(\frac{\pi}{2}-\delta_0) \tilde{v}_\perp^{\min} \leq \cos^K(\frac{\pi}{2}-\delta_0) v_\perp < \tilde{v}_\perp < \frac{v_\perp}{\cos^K(\frac{\pi}{2}-\delta_0)} \leq \frac{\tilde{v}_\perp^{\max}}{\cos^K(\frac{\pi}{2}-\delta_0)} $$

For simplicity, let us remove primes and tilde's for a moment and assume that $\gamma$ is a projected $(\delta_0, \Delta r_0, \Delta \varphi_0)$-regular path from $r$ back to $r$ such that the corresponding sample path given $v_\perp$ maps $(r,v_\perp)$ to $(r,v_\perp)$. The path $\gamma$ is easily produced by following several $(\delta_0, \Delta r_0, \Delta \varphi_0)$-regular segments forward and then the same segments backwards.
It is expected that if we perturb $\gamma$ keeping the return to $r$, the final value of the normal velocity $\hat{v}_\perp$ would deviate from the initial value $v_\perp$.
Repeating the perturbed path several times accumulates the deviation from $v_\perp$ with ultimate possibility to reach any values for the final normal velocity.

Note that if $\gamma$ consists of only $2$ segments the perturbation would not give $\hat{v}_\perp$ different from  $v_\perp$ since we have a restriction for the perturbed path to return to the initial point. Nevertheless $4$-segment paths have enough flexibility to produce the desired effect.

Each $4$-segment projected path $\tilde{\gamma}$ near $\gamma$ is uniquely determined by the initial position $r$ and the four angles $\varphi_1$, $\varphi_2$, $\varphi_3$, and $\varphi_4$ at which the particle reflects from the boundaries of the thermostats. For a path like that
$$f_r(\tilde{\gamma})=f_r(\varphi_1,\varphi_2,\varphi_3, \varphi_4)=\frac{\hat{v}_\perp}{v_\perp}=\frac{\cos(\varphi_1')\cos(\varphi_2')\cos(\varphi_3')\cos(\varphi_4')}{\cos(\varphi_1)\cos(\varphi_2)\cos(\varphi_3)\cos(\varphi_4)},$$

where primes correspond to the angles of incidences for the corresponding segments. The constraint we have on $f$ is that the path returns to the original point $r$. Note that $f_r(\gamma)=1$.

\begin{lemma} \label{lemma: variation in v_perp}
For any $r \in \partial \Gamma$ there exists $\alpha_{r}<1$ such that for any $\alpha \in (\alpha_{r},1]$, there exists a 4-segment projected $(\delta_0, \Delta r_0, \Delta \varphi_0)$-regular path from $r$ to $r$ with a property that for any corresponding sample path going from $(r,v_\perp)$ to $(r,\hat{v}_\perp)$, we have $\frac{\hat{v}_\perp}{v_\perp}=\alpha$. Similar path exists with $\frac{\hat{v}_\perp}{v_\perp}=\frac{1}{\alpha}$ property.
\end{lemma}

\begin{proof}
First we are going to construct $\hat{\gamma}$ such that $f_r(\hat{\gamma})=1$ and there exists $\tilde{\gamma}$ near $\hat{\gamma}$ such that $f_r(\tilde{\gamma}) \ne 1$.

Let $\gamma_1$ be a $1$-segment projected $(\delta_0,\Delta r_0,\Delta \varphi_0)$-regular path originating at $r$ and ending at some $r'$. Assume $r \in \partial D_i$ and $r' \in \partial D_j$. There exists a segment $\gamma_2$ that originates at $r$ and hits $D_j$ tangentially at some $r_{\parallel}'$ possibly intersecting other thermostats. Since $\gamma_1$ is $(\delta_0,\Delta r_0,\Delta \varphi_0)$-regular, there exists $r'' \in \partial D_i$ such that $\gamma' := r \to r' \to r'' \to r'_{\parallel} \to r$ is a $4$-segment path having the first two segments $(\delta_0,\Delta r_0,\Delta \varphi_0)$-regular and not intersecting any other thermostats. If $\hat{\gamma} := r \to r' \to r'' \to r' \to r$, then $f_r(\hat{\gamma})=1$ and $f_r(\gamma')=0$, implying that $f_r$ restricted to paths $r \to r' \to r'' \to \tilde{r} \to r$ is an analytic function of one variable $\tilde{r} \in D_i$ and is not constant on the connected component of $\hat{\gamma}$. Since $(\delta_0,\Delta r_0,\Delta \varphi_0)$-regularity is an open condition, there exists a projected $(\delta_0,\Delta r_0,\Delta \varphi_0)$-regular path $\tilde{\gamma} := r \to r' \to \tilde{r} \to r' \to r$ such that $f_r(\tilde{\gamma}) \ne 1$.

It may happen though that $f_r$ conditioned to return to $r$ has a critical point at $\hat{\gamma}$, in which case we are not guaranteed an open interval $[\alpha_r,1/\alpha_r]$ of $\frac{\hat{v}_\perp}{v_\perp}$ variation. However, since $f_r$ conditioned to return to $r$ is not constant, there exists a projected $(\delta_0,\Delta r_0,\Delta \varphi_0)$-regular path $\gamma$ near $\hat{\gamma}$ such that $f_r$ conditioned to return to $r$ does not have a critical point at $\gamma$. The conclusion of Lemma \ref{lemma: variation in v_perp} follows.

\end{proof}

\medbreak

Given $\alpha<1$, let $R_\alpha$  be the collection of initial positions $r \in \cup_i D_i$ such that for any $\tilde{\alpha} \in (\alpha, 1)$, there exist two 4-segment projected $(\delta_0, \Delta r_0, \Delta \varphi_0)$-regular paths from $r$ to $r$ with $\frac{\hat{v}_\perp}{v_\perp}=\alpha$ and $\frac{\hat{v}_\perp}{v_\perp}=\frac{1}{\alpha}$ respectively. Note that this time the $4$-segment paths are allowed to reflect off any thermostats. Each $R_\alpha$ is open and, by Lemma \ref{lemma: variation in v_perp} and continuity, $\{R_\alpha\}_{\alpha \in (0,1)}$ is a covering of the compact set $\cup_j D_j$. Extracting a finite sub-covering guarantees an upper bound $\alpha_r \leq \alpha_{\max}<1$ for all $r \in \cup_i D_i$.

Therefore for every $r \in \cup_i D_i$ and $\tilde{\alpha} \in (\alpha_{\max},\frac{1}{\alpha_{\max}})$, there exist two $4$-segment projected $(\delta_0, \Delta r_0, \Delta \varphi_0)$-regular paths along which $v_\perp$ changes by a fraction $\tilde{\alpha}$ and $\frac{1}{\tilde{\alpha}}$ respectively.

Let $n$  be such that $\alpha_{\max} < [ \frac{\tilde{v}_\perp^{\min} \cos^K(\frac{\pi}{2}-\delta_0)}{\tilde{v}_\perp^{\max}}]^\frac{1}{n}$ and $\gamma_K$ be the $K$-segment projected $(\delta_0, \Delta r_0, \Delta \varphi_0)$-regular path from $r$ to $r'$ guaranteed by Lemma \ref{lemma:path}. The corresponding sample path $\sigma_0$ given $\gamma_K$ and $v_\perp$ maps $(r,v_\perp)$ to $(r',\tilde{v}_\perp)$ for some $\tilde{v}_\perp$. Let
$$\tilde{\alpha}=[\min\{\frac{\tilde{v}_\perp}{v_\perp},\frac{v_\perp}{\tilde{v}_\perp}\}]^{\frac{1}{n}}$$.

Then $\alpha_{\max} \leq \tilde{\alpha} < 1$ since $\cos^K(\frac{\pi}{2}-\delta_0) \tilde{v}_\perp^{\min} < \tilde{v}_\perp < \frac{\tilde{v}_\perp^{\max}}{\cos^K(\frac{\pi}{2}-\delta_0)}$ and $\tilde{v}_\perp^{\min}
\leq v_\perp' \leq \tilde{v}_\perp^{\max}$. Let $\gamma_{\tilde{\alpha}}$ be the $4$-segment projected $(\delta_0, \Delta r_0, \Delta \varphi_0)$-regular path from $r'$ to $r'$ such that $\frac{\hat{v}_\perp}{v_\perp}=\tilde{\alpha}$ if $\tilde{v}_\perp<v_\perp $ or $\frac{\hat{v}_\perp}{v_\perp}=\frac{1}{\tilde{\alpha}}$ if $\tilde{v}_\perp>v_\perp$. Then if we extend $\gamma_K$ by $\gamma_{\tilde{\alpha}}$ repeated $n$ times we obtain a projected $(\delta_0, \Delta r_0, \Delta \varphi_0)$-regular path with the corresponding sample path $\sigma$ given $v_\perp$ going from $(r,v_\perp)$ to $(r',v_\perp')$ as desired. And the total number of collisions of such a path is $\leq M=K + 4n$.

\medbreak

\subsection{Proof of Proposition \ref{prop:sample path}}
Lemmas \ref{lemma:path} and \ref{lemma:v_perp boost} guarantee that there is a projected $(\delta_0, \Delta r_0, \Delta \varphi_0)$-regular path from $(r,v_\perp)$ to $(r',v_\perp')$ with $\leq M$  collisions provided that $\tilde{v}_\perp^{\min} \leq v_\perp,v_\perp' \leq \tilde{v}_\perp^{\max}$. We would like to find an $N \geq M$ such that for any choice of $(r,v_\perp)$ and $(r',v_\perp')$ we can construct a path that has exactly $N$ collisions.

Lemma \ref{lemma: angle property} provides us with a simple way to add two segments to a projected $(\delta_0, \Delta r_0, \Delta \varphi_0)$-regular path with return to the same point $r'$ with the same $v_\perp'$. However, there may be a parity issue if some of the paths from Lemma \ref{lemma:v_perp boost} have odd number of segments and others even.

For each \emph{odd} $k$, let $C_k$ be the set of all $r \in \cup_i \partial D_i$ such that there exists a projected $(\delta_0, \Delta r_0, \Delta \varphi_0)$-regular path from $r$ to back to $r$ that has $k$ collisions. Each $C_k$ is open.

\begin{claim} \label{claim2}
$\cup_k C_k$ is an open covering of $\cup_i \partial D_i$ provided that $\delta_0$, $\Delta r_0$, and $\Delta \varphi_0$ are small enough.
\end{claim}

\begin{proof}
First of all, we would like to choose $\delta_0$, $\Delta r_0$, and $\Delta \varphi_0$, smaller than before if necessary, so that $\cup_k C_k$ is non-empty. Let $r \in \partial D_i$ and $r' \in \partial D_j$ be the closest two points connected by a projected path, i.e. $|r'-r|=\tau_{\min}$. Then there exists a rectangle with two of the sides parallel to the segment $r \to r'$ that meets no other thermostats. Extend the sides perpendicular to $r \to r'$ until they meet some other thermostats $D_k$ and $D_l$ at points $r_k$ and $r_l$ respectively. If $r_k$ is at a corner of the rectangle, perturb $r \in D_i$, $r' \in D_j$ and/or $r_k \in D_k$ such that $\tilde{r} \to \tilde{r}' \to  \tilde{r}_k \to \tilde{r}$ is a $3$-segment projected path with all angles of incidence and refection not equal to $\pm \frac{\pi}{2}$; otherwise this property is guaranteed without the perturbation. Choosing $\delta_0$, $\Delta r_0$, and $\Delta \varphi_0$ smaller than before if necessary, we ensure that this path is $(\delta_0, \Delta r_0, \Delta \varphi_0)$-regular.

Assume now that some $r_1 \not \in \cup_j C_j$ and let $r_2 \in C_k$ for some $k$. Then Lemma \ref{lemma:path} provides a projected $(\delta_0, \Delta r_0, \Delta \varphi_0)$-regular path from $r_1$ to $r_2$ to which we append a path from $r_2$ to itself with $k$ thermostat collisions, and then the Lemma \ref{lemma:path} path in reverse direction, i.e from $r_2$ to $r_1$. This generates a projected $(\delta_0, \Delta r_0, \Delta \varphi_0)$-regular path from $r_1$ to itself with odd number of collisions with thermostats, a contradiction.
\end{proof}

Since $\cup_k C_k$ is an open covering of a compact set $\cup_i \partial D_i$, there exists a finite sub-covering $\{C_{k_1}, \cdots, C_{k_m}\}$ with $k_1 < k_2 < \cdots < k_m$.
Therefore for each $r \in \cup_i D_i$ there exits a $k_m$-segment projected $(\delta_0, \Delta r_0, \Delta \varphi_0)$-regular path $\gamma_r^m$.

Let us choose $N$ to be even. Then if the path from  $r$ to $r'$ given by Lemma \ref{lemma:path} contains odd number of segments, append it by $\gamma_{r'}^m$ to produce a projected $(\delta_0, \Delta r_0, \Delta \varphi_0)$-regular path with $\leq K+k_m$ segments. Then continue with boosting $v_\perp$ to the desired value $v_\perp'$ as in Lemma \ref{lemma:v_perp boost}. This produces a sample path $\sigma$ such that corresponding projected path has even number of segments and is $(\delta_0, \Delta r_0, \Delta \varphi_0)$-regular with $\leq N$ segments for some uniformly chosen even $N$. If such a path happens to make $<N$ thermostat collisions, we just add to it the required number of back and forth $2$-segment projected $(\delta_0, \Delta r_0, \Delta \varphi_0)$-regular paths guaranteed by Lemma \ref{lemma: angle property}.

\subsection{Pushing Density Forward: proof of Prop. \ref{prop: pushing density forward}} \label{subsect:push forward}
Given $\Delta r<\Delta r_0$ and $\Delta v_\perp < \frac{1}{2} \tilde{v}_\perp^{\min} \cos^N(\frac{\pi}{2}-\delta_0)$, let $\mu$ be the uniform probability measure in the neighborhood $(r-\Delta r,r+ \Delta r) \times (v_\perp-\Delta v_\perp,v_\perp+ \Delta v_\perp)$ of $(r,v_\perp) \in \tilde{\mathcal{C}}$. We would like to show that $\mathcal{P}^N_* \mu \geq \tilde{\eta}_N \tilde{\nu}$, where $\tilde{\nu}$ is the uniform probability measure on $\tilde{C}$.

The projected $(\delta_0, \Delta r_0, \Delta \varphi_0)$-regular path provided by Prop. \ref{prop:sample path} does not leave the set $\tilde{C}^{\max}= \{(r,v_\perp): \cos^N(\frac{\pi}{2}-\delta_0)\tilde{v}_\perp^{\min} \leq v_\perp \leq \frac{\tilde{v}_\perp^{\max}}{\cos^N(\frac{\pi}{2}-\delta_0)}\}$. In the following Lemma \ref{lemma: 1 push forward} we will give the desired bound for any two points in $\tilde{C}^{\max}$ that can be connected by a $1$-segment $(\delta_0, \Delta r_0, \Delta \varphi_0)$-regular path. Prop. \ref{prop: pushing density forward} follows by applying Lemma \ref{lemma: 1 push forward} $N$ times.

\begin{lemma} \label{lemma: 1 push forward}
Given $\Delta r<\Delta r_0$, $\Delta v_\perp < \frac{1}{2}  \tilde{v}_\perp^{\min} \cos^N(\frac{\pi}{2}-\delta_0)$ and $\eta>0$, there exist $\eta'$, $\Delta r'$, $\Delta v_\perp'$ such that if $(r,v_\perp), (r',v_\perp') \in \tilde{\mathcal{C}}^{\max}$ can be connected by a $1$-segment $(\delta_0, \Delta r_0, \Delta \varphi_0)$-regular path and $\mu$ is the uniform measure in the neighborhood $(r-\Delta r,r+ \Delta r) \times (v_\perp-\Delta v_\perp,v_\perp+ \Delta v_\perp)$ of $(r,v_\perp) \in \tilde{\mathcal{C}}$ having density $\eta$ with respect to Lebesgue measure $m$ on $\Omega$, then $\mathcal{P}_* \mu \geq \eta' \mu'$, where $\mu'$ is the restriction of $m$ to $(r'-\Delta r',r'+\Delta r') \times (v_\perp'-\Delta v_\perp',v_\perp'+\Delta v_\perp')$.
\end{lemma}

\begin{proof}

When we push $\mu$ forward by $\Phi$, we first perturb it in the angle direction after which we obtain a product measure in $3$-dimensional space with density $\eta \times \rho_{\tilde{v}_\perp}(\tilde{\varphi})$ at $(\tilde{r},\tilde{\varphi},\tilde{v}_\perp) \in (r-\Delta r,r+ \Delta r) \times (-\frac{\pi}{2},\frac{\pi}{2}) \times (v_\perp-\Delta v_\perp,v_\perp+ \Delta v_\perp)$. If we only allow $\varphi$ to vary by at most $\Delta \varphi_0 < \frac{\delta_0}{2}$, then the lower bound on this density is $\eta \times \rho_{\min}$, where $\rho_{\min}$ is the minimum value of $\rho_{\tilde{v}_\perp}(\tilde{\varphi})$ when $\cos^N(\frac{\pi}{2}-\delta_0)\tilde{v}_\perp^{\min} - \Delta v_\perp \leq \tilde{v}_\perp \leq \frac{\tilde{v}_\perp^{\max}}{\cos^N(\frac{\pi}{2}-\delta_0)}+\Delta v_\perp\}$ and $|\tilde{\varphi}| \leq \frac{\pi}{2}-\delta_0+\Delta \varphi_0$. Next we push this perturbed $3$-dimensional measure forward by enhanced billiard map i.e. the billiard map with $v_\perp$ as an additional variable. The derivative matrix for such a map is
\begin{equation} \label{eqn:enh der matrix}
\left(
      \begin{array}{lll}
        -\frac{\tau \kappa+\cos(\varphi)}{\cos(\varphi')}& -\frac{\tau}{\cos(\varphi')} & 0\\
        \frac{\tau \kappa \kappa'+\kappa \cos(\varphi')+\kappa' \cos(\varphi)}{\cos{\varphi'}} & \frac{\tau \kappa'+\cos(\varphi')}{\cos(\varphi')} & 0\\
        -\frac{\sin(\varphi')\frac{\partial \varphi'}{\partial r}}{\cos(\varphi)}v_\perp & \frac{-\sin(\varphi') \frac{\partial \varphi'}{\partial \varphi} \cos(\varphi)+\cos(\varphi')\sin(\varphi)}{\cos^2(\varphi)}v_\perp & \frac{\cos(\varphi')}{\cos(\varphi)}
      \end{array}
    \right)
\end{equation}
and has determinant $-1$ (note the cancelation of the usual billiard map derivative $\frac{\cos(\varphi)}{\cos(\varphi')}$ with the perpendicular velocity change $\frac{\cos(\varphi')}{\cos(\varphi)}$).
Therefore the lower bound on the density does not change when pushed forward. Now we need to project this pushed forward 3-dimensional measure back to $(r,v_\perp)$-plane ensuring certain lower bound on the density of the projected measure. The only thing we need to guarantee is

\begin{claim} \label{claim3}
There exist $\Delta r'$, $\Delta v_\perp'$, $\Delta \varphi'$ that do not depend on $(r,v_\perp)$ and $(r',v_\perp')$ such that the pre-image (under the enhanced billiard map) of the $\Delta r'  \times \Delta \varphi' \times \Delta v_\perp'$-neighborhood of $(r',\varphi', v_\perp')$ maps into $\Delta r  \times \Delta \varphi_0 \times \Delta v_\perp$-neighborhood of $(r,\varphi, v_\perp)$.
\end{claim}

Then the lower bound on the density of $\mathcal{P}_* \mu$ in $(r'-\Delta r, r' + \Delta r') \times (v_\perp'-\Delta v_\perp',v_\perp'+\Delta v_\perp')$ is $\eta'=\eta \times \rho_{\min}\times 2 \Delta \varphi'$.
\end{proof}

\medbreak

\begin{proof} [of Claim] Given $\Delta r'<\Delta r_0$ and $\Delta \varphi' < \Delta \varphi_0$, let $A$ be the set of all $1$-segment projected (backward) paths originating at $\tilde{r}'$ at angle $\tilde{\varphi}'$ with $(\tilde{r}',\tilde{\varphi}') \in (r'-\Delta r',r'+\Delta r') \times (\varphi'-\Delta \varphi',\varphi'+\Delta \varphi')$. We will denote the landing point of such a path by $(\tilde{r},\tilde{\varphi})$. Also let

$$G_1(\Delta r', \Delta \varphi')=\sup\limits_{A}|\tilde{r}-r|$$
$$G_2(\Delta r', \Delta \varphi')=\sup\limits_{A}|\tilde{\varphi}-\varphi|$$
$$G_3(\Delta r', \Delta \varphi')=\sup\limits_{A}|\frac{\cos(\tilde{\varphi})}{\cos(\tilde{\varphi}')}-\frac{\cos(\varphi)}{\cos(\varphi')}|$$
$$G_4(\Delta r', \Delta \varphi')=\sup\limits_{A}\frac{\cos(\tilde{\varphi})}{\cos(\tilde{\varphi}')}$$

Note that for $1 \leq i \leq 3$, $G_i(\Delta r',\Delta \varphi') \to 0$ uniformly as $(\Delta r', \Delta \varphi') \to 0$ as long as all angles are bounded away from $\frac{\pi}{2}$, which is ensured by the fact the $r' \to r$ is a projected $(\delta_0, \Delta r_0, \Delta \varphi_0)$-regular path. $G_4$ is bounded by $1 \leq G_4 \leq 1/\cos(\frac{\pi}{2}-\delta_0)$. Therefore we can always choose $\Delta r'$ and $\Delta \varphi'$ to ensure that
$$G_1(\Delta r', \Delta \varphi') < \Delta r$$
$$G_2(\Delta r', \Delta \varphi') < \Delta \varphi$$
$$G_3(\Delta r', \Delta \varphi') \tilde{v}^{\max}_\perp < \frac{\Delta v_\perp}{2}$$

Now
$$\sup\limits_A |\tilde{v}_\perp-v_\perp|=\sup\limits_A|\frac{\cos(\tilde{\varphi})}{\cos(\tilde{\varphi}')}\tilde{v}_\perp'-\frac{\cos(\varphi)}{\cos(\varphi')}v_\perp' | \leq G_3(\Delta r', \Delta \varphi') v_\perp' + G_4(\Delta r', \Delta \varphi')\Delta v_\perp'$$
and we want it to be less than $\Delta v_\perp$.
Choosing
$$\Delta v_\perp' < \frac{\Delta v_\perp}{2 G_4(\Delta r', \Delta \varphi')}$$
does the job.

Then for any $\tilde{r}' \in (r'-\Delta r', r'+\Delta r')$, $\tilde{\varphi}' \in (\varphi'-\Delta \varphi',\varphi' + \Delta \varphi')$, and $\tilde{v}_\perp' \in (v_\perp'-\Delta v_\perp', v_\perp' + \Delta v_\perp')$, we guarantee that the image of $(\tilde{r}',\tilde{\varphi}',\tilde{v}_\perp')$ under the enhanced billiard map belongs to the $\Delta r \times \Delta v_\perp \times \Delta \varphi$-neighborhood of $(r,\varphi, v_\perp)$.
\end{proof}

\subsection{Acquiring density: proof of Prop. \ref{prop: acquiring density}}

We would like to establish that there exist $\eta$, $\Delta r$, and $\Delta v_\perp$ such that for any $(r,v_\perp) \in \mathcal{C}$, $\mathcal{P}^2_* \delta_{(r,v_\perp)} \geq \eta m_{A_{(r,v_\perp)}}$, where $m$ is Lebesgue measure on $\Omega$ and $A_{(r,v_\perp)}=(r''-\Delta r,r''+\Delta r ) \times (v_\perp''-\Delta v_\perp, v_\perp''+\Delta v_\perp)$ for some $(r'',v_\perp'') \in \tilde{\mathcal{C}}$.

The $1$-step push forward measure $\mathcal{P}_* \delta_{(r,v_\perp)}$ is supported on a finite union of curves $L$ in the phase space $\Omega$ and is determined by the distribution $\rho_{v_\perp}(\varphi_1) d\varphi_1$ of the angle $\varphi_1$. It may happen that $(r,v_\perp)$ is a pre-image of $(r',v_\perp')$ for several different values of $\varphi_1$. Counting just one we conclude that the density at each $(r',v_\perp')$ is
$$  \geq \frac{1}{\sqrt{(\frac{\partial r'}{\partial \varphi_1})^2+(\frac{\partial v_\perp'}{\partial \varphi_1})^2}} \times \rho_{v_\perp}(\varphi_1), $$
where $\frac{\partial r'}{\partial \varphi_1}$ and $\frac{\partial v_\perp'}{\partial \varphi_1}$ are entries $(1,2)$ and $(3,2)$ respectively of the enhanced billiard derivative matrix (\ref{eqn:enh der matrix}).

The $2$-step push forward measure $\mathcal{P}^2_* \delta_{(r,v_\perp)}$ is, in addition, determined by the distribution $\rho_{v_\perp'}(\varphi_2)$ of $\varphi_2$, where $v_\perp'=\frac{\cos(\varphi_1')}{
\cos(\varphi_1)}v_\perp$. The density $d(\mathcal{P}^2_* \delta_{(r,v_\perp)})$ of $\mathcal{P}^2_* \delta_{(r,v_\perp)}$ at the endpoint $(r'',v_\perp'')$ of the sample path starting from $(r,v_\perp)$ and determined by $\varphi_1$ and $\varphi_2$ satisfies the estimate

$$d(\mathcal{P}_* \delta_{(r,v_\perp)}) \geq \rho_{v_\perp}(\varphi_1) \times \rho_{v_\perp'}(\varphi_2) \large/ |\det  \left(
  \begin{array}{cc}
    \frac{\partial r''}{\partial \varphi_1} & \frac{\partial r''}{\partial \varphi_2} \\
    \frac{\partial v_\perp''}{\partial \varphi_1}& \frac{\partial v_\perp''}{\partial \varphi_2}\\
  \end{array}
\right)|
,$$

where the derivatives are computed using (\ref{eqn:enh der matrix}) and a similar derivative matrix for the second flight. In particular, $\frac{\partial r''}{\partial \varphi_1}=\frac{\partial r''}{\partial r'}\frac{\partial r'}{\partial \varphi_1}$ and $\frac{\partial v_\perp''}{\partial \varphi_1}=\frac{\partial v_\perp''}{\partial r'}\frac{\partial r'}{\partial \varphi_1}$.

\begin{lemma}
There exists a uniform lower bound $\eta$ on $d(\mathcal{P}_*^2 \delta_{(r,v_\perp)})$ at any endpoint $(r'',v_\perp'')$ of a $2$-segment sample path for which the corresponding projected path is $(\delta_0, \Delta r_0,\Delta \varphi_0)$-regular and initial $v_\perp$ is such that $v_\perp^{\min} \leq v_\perp \leq v_\perp^{\max}$. In addition, the same bound holds at the corresponding endpoints if we allow the angles $\varphi_1$ and $\varphi_2$ vary within $\Delta \varphi_0$.
\end{lemma}

Note that at least one such path exists for any $r$ by Lemma \ref{lemma: angle property}.

\begin{proof}
By observing the entries of (\ref{eqn:enh der matrix}) we see that the determinant above is proportional to $v_\perp$, which is bounded above by $v_\perp^{\max}$. In addition, the terms in denominators are products of $\cos(\varphi_1)$, $\cos(\varphi_1')$, $\cos(\varphi_2)$, and/or $\cos(\varphi_2')$ of total power no greater than $3$. Therefore by $(\delta_0, \Delta r_0,\Delta \varphi_0)$-regularity of the path, the determinant is bounded above by some constant $C=C(\delta_0, v_\perp^{\max})$.

In addition, if we let $\varphi_1$ and $\varphi_2$ vary by at most $\Delta \varphi_0$, we get a similar bounding constant. The product $\rho_{v_\perp}(\varphi_1) \times \rho_{v_\perp'}(\varphi_2)$ is clearly bounded for angle variations within $\Delta \varphi_0$ and $v_\perp^{\min} \leq v_\perp \leq v_\perp^{\max}$. The conclusion follows.
\end{proof}

\medbreak

Unfortunately $\Delta \varphi_0$ variation in $\varphi_1$ and $\varphi_2$ does not always produce nonzero variation in both $\Delta r''$ and $\Delta v_\perp''$. For example, if we start from any location $r \in \partial \Gamma$ and choose $r''=r$, then we get no variation in $v_\perp''$ at $r''$. Note that
$\det \tiny
\left( \begin{array}{cc}
\frac{\partial r''}{\partial \varphi_1} & \frac{\partial r''}{\partial \varphi_2} \\
    \frac{\partial v_\perp''}{\partial \varphi_1}& \frac{\partial v_\perp''}{\partial \varphi_2}\\
\end{array}
       \right)=0$ \normalsize in this case. Indeed, such paths correspond to critical points of $v_\perp''(\varphi_1,\varphi_2)$ constrained by $r''(\varphi_1,\varphi_2)=const$. In order to guarantee nonzero variation in $r''$ and $v_\perp''$, for each $r$, we must find a $2$-segment projected $(\delta_0, \Delta r_0,\Delta \varphi_0)$-regular path such that the above determinant is uniformly bounded away from zero.

\begin{lemma} \label{lemma:nonzero det}
For any $r$ there exists a $2$-segment projected $(\delta_0,\Delta r_0,\Delta \varphi_0)$-regular path $\gamma$ defined by angles of reflection $\varphi_1$ and $\varphi_2$ with $|J(r,\varphi_1,\varphi_2)|>0$, where $J(r,\varphi_1,\varphi_2):=\det  \tiny \left( \begin{array}{cc}
\frac{\partial r''}{\partial \varphi_1} & \frac{\partial r''}{\partial \varphi_2} \\
    \frac{\partial v_\perp''}{\partial \varphi_1}& \frac{\partial v_\perp''}{\partial \varphi_2}\\
                                   \end{array}
                                 \right) \normalsize /v_\perp$.
\end{lemma}

\begin{proof}
Fix the first segment of $\gamma$ to be any $1$-segment projected $(\delta_0,\Delta r_0,\Delta \varphi_0)$-regular path ending at some $r'$ with angle of reflection $\varphi_1$. Then $J(r,\varphi_1,\varphi_2)$ is an analytic function of one variable $\varphi_2$. Similar to the proof of Lemma \ref{lemma: variation in v_perp} there exist two $1$-segment projected paths $\gamma_1$ and $\gamma_2$ originating at $r'$ and ending at the same thermostat such that $\gamma_1$ is $(\delta_0,\Delta r_0,\Delta \varphi_0)$-regular and $\gamma_2$ has the angle of incidence $\varphi=\pm \frac{\pi}{2}$. Therefore, $J(r,\varphi_1,\varphi_2)$ is not constant on the connected component and there exists $\varphi_2$ such that $J(r,\varphi_1,\varphi_2)>0$.
\end{proof}

For any $r$, let $f(r)=\sup\{J(r,\varphi_1,\varphi_2)\}$, where $\sup$ is taken along all $2$ segment projected $(\delta_0,\Delta r_0,\Delta \varphi_0)$-regular paths. $f$ is an everywhere positive continuous function on a compact set and therefore there exists the minimum value $f_0>0$. For each $r$, fix a $2$-segment projected $(\delta_0,\Delta r_0,\Delta \varphi_0)$-regular path $\gamma(r)$ given by $\overline{\varphi}_1(r)$ and $\overline{\varphi}_2(r)$ such that $|J(r,\overline{\varphi}_1(r),\overline{\varphi}_2(r))|>\frac{f_0}{2}$. By the Lagrange multipliers method $r''(\varphi_1,\varphi_2)$ constrained by $v_\perp''(\varphi_1,\varphi_2)/v_\perp=const$ does not have a critical point at $(\overline{\varphi}_1(r), \overline{\varphi}_2(r))$ and $v_\perp''(\varphi_1,\varphi_2)/v_\perp$ constrained by $r''(\varphi_1,\varphi_2)=const$ does not have a critical point at $(\overline{\varphi}_1(r), \overline{\varphi}_2(r))$.

Let $\epsilon(r)>0$ be the maximal allowed variation such that any $(\tilde{r}'',\tilde{v}_\perp'') \in (r''-\epsilon, r''+\epsilon) \times (v_\perp''-v_\perp\epsilon,v_\perp''+v_\perp\epsilon)$ is an endpoint of a sample path given by $(r, \tilde{\varphi}_1, \tilde{\varphi}_2, v_\perp)$ with $\tilde{\varphi}_1 \in (\overline{\varphi}_1(r) - \Delta \varphi_0, \overline{\varphi}_1(r) + \Delta \varphi_0)$, $\tilde{\varphi}_2 \in (\overline{\varphi}_2(r) - \Delta \varphi_0, \overline{\varphi}_2(r) + \Delta \varphi_0)$, and $J(r,\tilde{\varphi}_1(r),\tilde{\varphi}_2(r))>\frac{f_0}{2}$. $\epsilon(r)$ is a continuous positive function on a compact domain $\partial \Gamma$ and achieves its minimum somewhere. Let $\epsilon=\min\limits_r\epsilon(r)$.

Then for each $(r,v_\perp)$ and $r''$ and $v_\perp''$ given by $\gamma(r)$, at each $(\tilde{r}'',\tilde{v}_\perp'') \in A_{r,v_\perp} = (r''-\epsilon, r''+\epsilon)\times  (v_\perp''-v_\perp^{\min}\epsilon,v_\perp''+v_\perp^{\min}\epsilon)$, $d(\mathcal{P}_* \delta_{(r,v_\perp)}) \geq \eta$.

\medbreak

\subsection{Proof of absolute continuity of the invariant measure with respect to Lebesgue measure $m$ on $\Omega$}
\begin{lemma} \label{lemma:absolute continuity}
If $\mu$ is invariant under $\Phi_\tau$ and for all $(r,v_\perp) \in \Omega$ the absolutely continuous component of $\mathcal{P}^2_* \delta_{(r,v_\perp)}$ is nonzero, then $\mu \ll m$.
\end{lemma}

The assumption of Lemma \ref{lemma:absolute continuity} follows from the proof of Prop. \ref{prop: acquiring density}: for each $(r,v_\perp) \in \Omega$, we constructed a projected $(\delta_0,\Delta r_0,\Delta \varphi_0)$-regular path in some neighborhood of which the pushed forward measure is absolutely continuous with respect to $m$ (we do not need a uniform lower bound on the density here).

\begin{proof}
Assume $\mu$ is not absolutely continuous with respect to $m$. Then $\mu$ can be decomposed as a sum of absolutely continuous and singular components $\mu=\mu_{\ll}+\mu_{\perp}$ with $\mu_{\perp}(\Omega) > 0$. By the invariance of $\mu$
\begin{equation} \label{eqn:abs cont}
\mu_{\ll}+\mu_{\perp}=\mu=\mathcal{P}^2_* \delta_{(r,v_\perp)}=(\mathcal{P}^2_* \mu_{\ll})_{\ll}+(\mathcal{P}^2_* \mu_{\ll})_{\perp}+(\mathcal{P}^2_* \mu_{\perp})_{\ll}+(\mathcal{P}^2_* \mu_{\perp})_{\perp}
\end{equation}

Any push forward of absolutely continuous measure is absolutely continuous, which follows from the proof of Prop. \ref{prop: pushing density forward} if we drop the regularity bounds and only keep track of absolute continuity. Therefore $(\mathcal{P}^2_* \mu_{\ll})_{\perp}$=0 and $\mu_{\ll}=(\mathcal{P}^2_* \mu_{\ll})_{\ll}$. In addition, Prop. \ref{prop: acquiring density} guarantees that $(\mathcal{P}^2_* \mu_{\perp})_{\perp}>0$, since we assumed that $\mu_{\perp}(\Omega) > 0$. Therefore
$$\mu_{\ll}(\Omega)<(\mathcal{P}^2_* \mu_{\ll})_{\ll}(\Omega)+(\mathcal{P}^2_* \mu_{\perp})_{\ll}(\Omega),$$

which is a contradiction since absolutely continuous part on the left hand side of (\ref{eqn:abs cont}) must be equal to the absolutely continuous part on the right hand side.
\end{proof}

\subsection{Proof of Lemma \ref{lemma: equilibrium}}
Let us first lift the density $\frac{2 \beta}{|\partial \Gamma|} v_\perp e^{-\beta v^2_\perp}$ to the $3$-dimensional space by sampling $\varphi$ from $\rho_{v_\perp}(\varphi)$: the resulting density at $(r,\varphi,v_\perp)$ is
$$2 \beta v_\perp e^{-\beta v^2_\perp} \rho_{v_\perp}(\varphi).$$

We push the measure with this density forward under the enhanced billiard map, the Jacobian of this map is equal to $1$ (see Formula (\ref{eqn:enh der matrix})), so the value of the density of the pushed forward measure does not change at the corresponding new location $(r',\varphi',v_\perp')$.

To find the value of the density of the pushed forward measure at $(r',v_\perp')$  we just need to integrate over all possible values at $(r,\varphi, v_\perp)$ that map to $(r',\varphi',v_\perp')$ as $\varphi'$ ranges over $(-\frac{\pi}{2},\frac{\pi}{2})$. Note that $v_\perp'=\frac{\cos(\varphi')}{\cos(\varphi)}$. Then we obtain

$$d(\mathcal{P}_* \mu) = \int\limits_{-\frac{\pi}{2}}^{\frac{\pi}{2}} 2 \beta v_\perp e^{-\beta v^2_\perp} \sqrt{\frac{\beta}{\pi}} \frac{v_\perp}{\cos^2(\varphi)} e^{-\beta v^2_\perp \tan^2(\varphi)} d \varphi'$$
$$=\int\limits_{-\frac{\pi}{2}}^{\frac{\pi}{2}}  2 \beta \sqrt{\frac{\beta}{\pi}} \frac{(v_\perp')^2}{\cos^2(\varphi')}e^{-\beta (v_\perp')^2 / \cos^2(\varphi')} d\varphi'=2 \beta v_\perp' e^{-\beta (v_\perp')^2} \int\limits_{-\infty}^{\infty} \frac{1}{\sqrt{\pi}} e^{-u^2} du$$
$$=2 \beta v_\perp' e^{-\beta (v_\perp')^2}$$
as desired. In the computation used a change of variables $u=\sqrt{\beta} v_\perp' \tan(\varphi')$.

\vskip 0.2in

\textbf{Acknowledgement:}\\
Tatiana Yarmola would like to thank her Ph.D. thesis advisor Lai-Sang Young for help with formulation of the problem, fruitful discussions, effective criticism, and useful comments. This work was partially supported by the National Science Foundation Postdoctoral Research Fellowship.

\end{document}